\documentclass[onecolumn,12pt,draftclsnofoot]{IEEEtran}

\usepackage{amsmath,amssymb,amsthm}
\usepackage{graphicx}
\usepackage{subfigure}
\usepackage{curves}
\usepackage{amsfonts}
\usepackage{epsfig}
\usepackage{stfloats}
\usepackage{cite}
\usepackage[ruled,vlined]{algorithm2e}
\usepackage{epstopdf}
\usepackage{multirow}
\usepackage{array}
\usepackage{rotating}
\usepackage{comment}

\begin{document}

\title{Polar Coding Strategies for the Interference Channel with Partially-Joint Decoding}
\author{Mengfan~Zheng, Cong~Ling, Wen~Chen and Meixia~Tao
\thanks{M. Zheng, W. Chen and M. Tao are with the Department of Electronic Engineering at Shanghai Jiao Tong University, Shanghai, China. Emails: \{zhengmengfan, wenchen, mxtao\}@sjtu.edu.cn. C. Ling is with the Department of Electrical and Electronic Engineering at Imperial College London, United Kingdom. Email: c.ling@imperial.ac.uk.

The corresponding author is M. Tao.
}
}

\maketitle

\begin{abstract}
  Existing polar coding schemes for the two-user interference channel follow the original idea of Han and Kobayashi, in which component messages are encoded independently and then mapped by some deterministic functions (i.e., homogeneous superposition coding). In this paper, we propose a new polar coding scheme for the interference channel based on the heterogeneous superposition coding approach of Chong, Motani and Garg. We prove that fully-joint decoding (the receivers simultaneously decode both senders' common messages and the intended sender's private message) in the Han-Kobayashi strategy can be simplified to two types of partially-joint decoding, which are friendly to polar coding with practical decoding algorithms. The proposed coding scheme requires less auxiliary random variables and no deterministic functions, and can be efficiently constructed. Further, we extend this result to interference networks and show that the proposed partially-joint decoding scheme is a general method for designing heterogeneous superposition polar coding schemes in interference networks.
  
\end{abstract}


\section{Introduction}

  Polar codes, proposed by Ar{\i}kan \cite{arikan2009channel}, are the first class of channel codes that can provably achieve the capacity of any memoryless binary-input output-symmetric channels with low encoding and decoding complexity. Since its invention, polar codes have been widely adopted to many other scenarios, such as source compression \cite{arikan2010source,korada2009polar,korada2010lossy,arikan2012sw}, wiretap channels \cite{andersson2010nested,mahdavifar2011achieving,koyluoglu2012polar,sasoglu2013strong,gulcu2015,wei2016generalwt}, relay channels \cite{andersson2010nested,karzand2012polar,blasco2012relay}, multiple access channels (MAC) \cite{arikan2012sw,onayscmac,sasoglu2013mac,abbe2012mmac,mahdavifar2016uniform}, broadcast channels \cite{goela2015broadcast,mondelli2015marton}, broadcast channels with confidential messages \cite{gulcu2015,chou2016broadcast}, and bidirectional broadcast channels with common and confidential messages \cite{andersson2013broad}. In these scenarios, polar codes have also shown capacity-achieving capabilities.
  
  The interference channel (IC), first initiated by Shannon \cite{shannon1961two} and further studied by Ahlswede \cite{ahlswede1974capacity}, models the situation where $m$ sender-receiver pairs try to communicate simultaneously through a common channel. In this model, it is assumed that there is no cooperation between any of the senders or receivers, and the signal of each sender is seen as interference by the unintended receivers. Although the 2-user discrete memoryless IC (DM-IC) is rather simple in appearance, except for some special cases \cite{Carleial1975case,sato1981gaussian,chung2008frequency,liu2008degraded,benzel1979additive,chong2009semi,gamal1982determin,costa1987strong}, determining the capacity region of a general IC remains an open problem. Reference \cite{ahlswede1974capacity} gave simple but fundamental inner and outer bounds on the capacity region of the IC. In \cite{Carleial1978IC}, Carleial determined an improved achievable rate region for the IC by applying the superposition coding technique of Cover \cite{cover1975broadcast}, which was originally designed for the broadcast channel. Later, Han and Kobayashi established the best achievable rate region for the general IC to date \cite{han1981ic}. A more compact description of the Han-Kobayashi region was given in \cite{chong2008han}. The idea of the Han-Kobayashi coding strategy is to split each sender's message into a private part and a common part, and allow the unintended receiver to decode the common part so as to enhance the total transmission rates. To achieve the whole Han-Kobayashi region, it is required that each receiver decodes its intended private message and both senders' common messages jointly.
  
  There are limited studies on the design of specific coding schemes that can achieve the Han-Kobayashi region. A low-density parity-check (LDPC) code-based Han-Kobayashi scheme was proposed for the Gaussian IC in \cite{sharifi2015ldpcic}, which has close-to-capacity performance in the case of strong interference. In \cite{sharifi2015ldpcz}, a specific coding scheme was designed for the binary-input binary-output Z IC using LDPC codes, and an example was shown to outperform time sharing of single user codes. For polar codes, reference \cite{appaiah2011align} pointed out how alignment of polarized bit-channels can be of use for designing coding schemes for interference networks, and presented an example of the one-sided discrete memoryless 3-user IC with a degraded receiver structure. A polar coding scheme that achieves the Han-Kobayashi inner bound for the 2-user IC was proposed in \cite{wang2015channel}, and \cite{Hirche2016quantum} used a similar scheme to achieve the Han-Kobayashi region in the 2-user classical-quantum IC. The idea of \cite{wang2015channel} is to transform the original IC into two 3-user MACs from the two receivers' perspectives, and design a compound MAC polar coding scheme for them. The achievable rate region of the compound MAC equals the Han-Kobayashi region, and can be achieved by polar codes. This design is based on the original Han-Kobayashi scheme of \cite{han1981ic}, in which component messages are independently encoded into auxiliary sequences and then mapped to the channel inputs by some deterministic functions (also known as \textit{homogeneous superposition coding} \cite{wang2013compare}). By ranging over all possible choices of these functions and distributions of auxiliary random variables (ARV), the whole Han-Kobayashi region can be achieved. However, such an approach could be problematic in practice since finding such functions may be a very complex task.
  
  Our work is inspired by the compact description of the Han-Kobayashi region based on the Chong-Motani-Garg scheme \cite{chong2008han}, in which no deterministic functions are required and less ARVs are needed. This approach belongs to the \textit{heterogeneous superposition coding} scheme \cite{wang2013compare}, in which the common message is encoded first and then a satellite codebook for the private message is generated around it. When implementing such a scheme using polar codes, we find that the \textit{fully-joint decoder} which simultaneously decodes all three component messages is difficult to design, because the encoding scheme forces us to decode the common message of a sender before its private message when successive cancellation decoding (SCD) is used. By analyzing points on the dominant faces of the Han-Kobayashi region and utilizing random coding techniques, we find that it is possible to loosen the fully-joint decoding requirement and propose to use two types of \textit{partially-joint decoders}. Each receiver can either jointly decode both senders' common messages first and then the intended sender's private message, or solely decode the intended sender's common message first and then jointly decode the rest two. Based on this finding and enlightened by Goela et al.'s superposition polar coding scheme for the broadcast channel \cite{goela2015broadcast}, we design two types of polar coding schemes and show that every point on the dominant faces of the Han-Kobayashi region can be achieved. Compared with the existing scheme of \cite{wang2015channel}, our proposed scheme achieves a larger rate region for the same joint distribution of random variables and can be constructed efficiently. Most notably, with the proposed scheme, the task of finding proper ARVs for a DM-IC can be reduced significantly. Further, we extend the partially-joint decoding scheme to arbitrary discrete memoryless interference networks (DM-IN) and show that it is a general method for designing heterogeneous superposition polar coding schemes in DM-INs that can achieve optimal rate regions. 
  
  In our proposed scheme, joint decoders and the corresponding code structure are implemented using the 2-user MAC polarization method based on Ar{\i}kan's monotone chain rule expansions \cite{arikan2012sw}, whose encoding and decoding complexity is similar to single-user polar codes. Besides, we propose a constructing method for this kind of MAC polar codes based on the approximation method of \cite{tal2013construct}, which makes our proposed codes easy to be constructed. We use \c{S}a\c{s}o\u{g}lu's result on polarization for arbitrary discrete alphabet \cite{sasoglu2011polar} to extend it to arbitrary prime input alphabet case.
  To deal with non-uniform input distribution, one may apply Gallager's alphabet extension method \cite[p. 208]{gallager1968information} as in \cite{wang2015channel}, the chaining construction \cite{Mondelli2014asymmetric}, or a more direct approach by invoking results on polar coding for lossless compression \cite{honda2013asymmetric,goela2015broadcast,chou2015deterministic,chou2016broadcast}. In this paper, we take Chou and Bloch's low-complexity approach \cite{chou2015deterministic,chou2016broadcast}, which only requires a vanishing rate of shared randomness between communicators. One crucial point in designing capacity-achieving polar codes for a general multi-user channel is how to properly align the polar indices. One solution for this problem is the chaining method, which has already been used in several areas \cite{sasoglu2013strong,hassani2014universal,gulcu2015,wei2016generalwt,mondelli2015marton}. Another way is to add additional stages of polarization to align the incompatible indices, as shown in \cite{sasoglu2016universal} and used in \cite{wang2015channel}. In this paper, we adopt the chaining method as it does not change the original polar transformation and may be easier to understand.
  
  The rest of this paper is organized as follows. In Section \ref{S:ProblemState}, we introduce the 2-user DM-IC model and the Han-Kobayashi region, and propose two types of partially-joint decoders. In Section \ref{S:PolarPri}, we review some background on polarization and polar codes necessary for our code design. In Section \ref{S:Overview}, we provide an overview of our scheme and analyze its feasibility. Details of our proposed schemes are presented in Section \ref{S:PolarScheme}, and the performance is analyzed in Section \ref{S:PA}. In Section \ref{S:IN}, we extend the proposed scheme to arbitrary DM-INs. Section \ref{S:Conclusion} concludes this paper with some discussions.

  \textit{Notations:} $[N]$ is the abbreviation of an index set $\{1,2,...,N\}$. Vectors are denoted as $\mathbf{X}^N\triangleq\{X^1,X^2,...,X^N\}$ or $X^{a:b}\triangleq\{X^a,X^{a+1},...,X^{b}\}$ for $a\leq b$. For a subset $\mathcal{A}\subset [N]$, $X^{\mathcal{A}}$ denotes the subvector $\{X^i:i\in\mathcal{A}\}$ of $X^{1:N}$. $\mathbf{G}_N=\mathbf{B}_N \textbf{F}^{\otimes n}$ is the generator matrix of polar codes \cite{arikan2009channel}, where $N=2^n$ with $n$ being an arbitrary integer, $\mathbf{B}_N$ is the bit-reversal matrix, and $\textbf{F}=
  \begin{bmatrix}
  1 & 0 \\
  1 & 1
  \end{bmatrix}$. $H_{q}(X)$ stands for the entropy of $X$ with $q$-based logarithm, and $H(X)$ is short for the Shannon entropy unless otherwise specified. $\delta_N=2^{-N^\beta}$ with some $\beta \in (0,1/2)$.

\section{Problem Statement}
\label{S:ProblemState}
 \subsection{Channel Model}

 \newtheorem{definition}{Definition}
 \begin{definition}
 	\label{def:IC}
 	A 2-user DM-IC consists of two input alphabets $\mathcal{X}_1$ and $\mathcal{X}_2$, two output alphabets $\mathcal{Y}_1$ and $\mathcal{Y}_2$, and a probability transition function $P_{Y_1Y_2|X_1X_2}(y_1,y_2|x_1,x_2)$. The conditional joint probability distribution of the 2-user DM-IC over $N$ channel uses can be factored as
 	\begin{equation}
 	P_{\mathbf{Y}_1^N\mathbf{Y}_2^N|\mathbf{X}_1^N\mathbf{X}_2^N}(\mathbf{y}_1^N,\mathbf{y}_2^N|\mathbf{x}_1^N,\mathbf{x}_2^N)=\prod_{i=1}^{N}P_{Y_1Y_2|X_1X_2}(y_1^i,y_2^i|x_1^i,x_2^i).
 	\end{equation}
 \end{definition}
 
 \begin{definition}
 	A $(2^{NR_1},2^{NR_2},N)$ code for the 2-user DM-IC consists of two message sets $\mathcal{M}_1=\{1,2,...,[2^{NR_1}]\}$ and $\mathcal{M}_2=\{1,2,...,[2^{NR_2}]\}$, two encoding functions
 	\begin{equation}
 	x_1^N(m_1):\mathcal{M}_1\mapsto \mathcal{X}_1^N \text{  and  } x_2^N(m_2):\mathcal{M}_2\mapsto \mathcal{X}_2^N,
 	\end{equation}
 	and two decoding functions
 	\begin{equation}
 	\hat{m}_1(\mathbf{y}_1^N):\mathcal{Y}_1^N\mapsto \mathcal{M}_1 \text{  and  } \hat{m}_2(\mathbf{y}_2^N):\mathcal{Y}_2^N\mapsto \mathcal{M}_2.
 	\end{equation}
 \end{definition}
 
 \begin{definition}
 	The average probability of error $P_e^{(N)}$ of a $(2^{NR_1},2^{NR_2},N)$ code for the 2-user DM-IC is defined as the probability that the decoded message pair is not the same as the transmitted one averaged over all possible message pairs,
 	\begin{equation}
 	P_e^{(N)}=\frac{1}{2^{N(R_1+R_2)}} \sum_{(M_1,M_2)\in\mathcal{M}_1\times \mathcal{M}_2}\mathrm{Pr}\Big{\{}\big{(}\hat{m}_1(\mathbf{Y}_1^N),\hat{m}_2(\mathbf{Y}_2^N)\big{)}\neq (M_1,M_2)|(M_1,M_2)~\text{sent}\Big{\}},
 	\end{equation}
 	where $(M_1,M_2)$ are assumed to be uniformly distributed over $\mathcal{M}_1\times \mathcal{M}_2$.
 \end{definition}
  
 \subsection{The Han-Kobayashi Rate Region}
 \label{S:HKRIntro}
 In the Han-Kobayashi coding strategy, each sender's message is split into two parts: a private message, which only needs to be decoded by the intended receiver, and a common message, which is allowed to be decoded by the unintended receiver. Each receiver decodes its intended private message and two common messages jointly so that a higher transmission rate can be achieved. In the rest of this paper, we will refer to the two senders and two receivers as Sender 1, Sender 2, Receiver 1 and Receiver 2 respectively. Sender 1's message, denoted as $M_1$, is split into $(M_{1p},M_{1c})$, where $M_{1p}\in\mathcal{M}_{1p}\triangleq\{1,2,...,[2^{NS_1}]\}$ denotes its private message and $M_{1c}\in \mathcal{M}_{1c}\triangleq\{1,2,...,[2^{NT_1}]\}$ the common message. Similarly, Sender 2's message $M_2$ is split into $(M_{2p},M_{2c})$ with $M_{2p}\in\mathcal{M}_{2p}\triangleq\{1,2,...,[2^{NS_2}]\}$ and $M_{2c}\in\mathcal{M}_{2c}\triangleq\{1,2,...,[2^{NT_2}]\}$. Define $W_1$, $W_2$, $V_1$ and $V_2$ as the random variables for messages $M_{1c}$, $M_{2c}$, $M_{1p}$ and $M_{2p}$ respectively, with $\mathcal{W}_1$, $\mathcal{W}_2$, $\mathcal{V}_1$ and $\mathcal{V}_2$ being their alphabets. Then each encoding function can be decomposed into three functions. For $x_1^N(m_1)$, the three functions are
 \begin{equation}
 \label{HKStr-1}
 \begin{aligned}
 &w_1^N(M_{1c}):\mathcal{M}_{1c}\mapsto \mathcal{W}_1^N,~~v_1^N(M_{1p}):\mathcal{M}_{1p}\mapsto \mathcal{V}_1^N\\
 &~~~~~~\text{and } x_1^{'N}(\mathbf{W}_1^N,\mathbf{V}_1^N):\mathcal{W}_1^N\times\mathcal{V}_1^N\mapsto \mathcal{X}_1^N.
 \end{aligned}
 \end{equation}
 Similarly, for $x_2^N(m_2)$, the three functions are
 \begin{equation}
 \label{HKStr-2}
 \begin{aligned}
 &w_2^N(M_{2c}):\mathcal{M}_{2c}\mapsto \mathcal{W}_2^N,~~v_2^N(M_{2p}):\mathcal{M}_{2p}\mapsto \mathcal{V}_2^N\\
 &~~~~~~\text{and } x_2^{'N}(\mathbf{W}_2^N,\mathbf{V}_2^N):\mathcal{W}_2^N\times\mathcal{V}_2^N\mapsto \mathcal{X}_2^N.
 \end{aligned}
 \end{equation}
 
 With this approach, Han and Kobayashi established the best achievable rate region for the general IC to date \cite{han1981ic}. The result is summarized in Theorem \ref{theorem:HK-o}. 
  \newtheorem{theorem}{Theorem}
  \begin{theorem}[\cite{han1981ic,el2011network}]
  	\label{theorem:HK-o}
  	
  	Let $\mathcal{P}^*$ be the set of probability distributions $P^*(\cdot)$ that factor as
  	\begin{equation}
  	\begin{aligned}
  	P^*(q,v_1,v_2,w_1,w_2,x_1,x_2)&=P_Q(q)P_{V_1|Q}(v_1|q)P_{V_2|Q}(v_2|q)P_{W_1|Q}(w_1|q)P_{W_2|Q}(w_2|q)\\
  	&~~\times P_{X_1|V_1W_1Q}(x_1|v_1,w_1,q)P_{X_2|V_2W_2Q}(x_2|v_2,w_2,q),
  	\end{aligned}
  	\end{equation}
  	where $Q\in \mathcal{Q}$ is the time-sharing parameter, and $P_{X_1|V_1W_1Q}(\cdot)$ and $P_{X_2|V_2W_2Q}(\cdot)$ equal either 0 or 1, i.e., they are deterministic functions.
  	For a fix $P^*(\cdot)\in \mathcal{P}^*$, consider Receiver 1 and the set of non-negative rate-tuples $(S_1,T_1,S_2,T_2)$ denoted by $\mathcal{R}_{HK}^{o,1}(P^*)$ that satisfy
  	\begin{align}
  	0\leq S_1 &\leq I(V_1;Y_1|W_1W_2Q), \label{HK1-1}\\
  	0\leq T_1 &\leq I(W_1;Y_1|V_1W_2Q), \label{HK1-2}\\
  	0\leq T_2 &\leq I(W_2;Y_1|V_1W_1Q), \label{HK1-3}\\
  	S_1+T_1 &\leq I(V_1W_1;Y_1|W_2Q), \label{HK1-4}\\
  	S_1+T_2 &\leq I(V_1W_2;Y_1|W_1Q), \label{HK1-5}\\
  	T_1+T_2 &\leq I(W_1W_2;Y_1|V_1Q), \label{HK1-6}\\
  	S_1+T_1+T_2 &\leq I(V_1W_1W_2;Y_1|Q). \label{HK1-7}
  	\end{align}
  	Similarly, let $\mathcal{R}_{HK}^{o,2}(P^*)$ be the set of non-negative rate-tuples $(S_1,T_1,S_2,T_2)$ that satisfy (\ref{HK1-1})--(\ref{HK1-7}) with indices 1 and 2 swapped everywhere. For a set $\mathcal{S}$ of 4-tuples $(S_1,T_1,S_2,T_2)$, let $\mathcal{R}(\mathcal{S})$ be the set of $(R_1,R_2)$ such that $0\leq R_1\leq S_1+T_1$ and $0\leq R_2\leq S_2+T_2$ for some $(S_1,T_1,S_2,T_2)\in \mathcal{S}$. Then we have that
  	\begin{equation}
  	\mathcal{R}_{HK}^o=\mathcal{R}\Big{(}\bigcup_{P^*\in \mathcal{P}^*}\mathcal{R}_{HK}^{o,1}(P^*)\cap \mathcal{R}_{HK}^{o,2}(P^*)\Big{)}
  	\end{equation}
  	is an achievable rate region for the DM-IC.
  \end{theorem}
  
  The original Han-Kobayashi scheme can be classified into the homogeneous superposition coding scheme \cite{wang2013compare}, in which the component messages of each sender are independently encoded into auxiliary sequences and then mapped to the channel input sequence by some symbol-by-symbol deterministic function. The scheme of \cite{wang2015channel} belongs to the this type. Another variant of superposition coding is the heterogeneous superposition coding \cite{wang2013compare}, introduced by Bergmans \cite{bergmans1973broad}. In this variant, the coarse messages are encoded into auxiliary sequences first, and then a satellite codebook for the fine message is generated around it conditionally independently. Usually the heterogeneous variant is simpler than the homogeneous one since it requires fewer ARVs. Reference \cite{chong2008han} presented a simplified description of Han-Kobayashi region based on this approach (referred to as the Chong-Motani-Garg scheme in this paper), in which only three ARVs are used and no deterministic functions are needed. Their result is summarized in Theorem \ref{theorem:HK}.
  \begin{theorem}[\cite{chong2008han,el2011network}]
  \label{theorem:HK}
  Let $\mathcal{P}_1^*$ be the set of probability distributions $P_1^*(\cdot)$ that factor as
  \begin{equation}
  P_1^*(q,w_1,w_2,x_1,x_2)=P_Q(q)P_{X_1W_1|Q}(x_1,w_1|q)P_{X_2W_2|Q}(x_2,w_2|q),
  \end{equation}
  where $|\mathcal{W}_j|\leq |\mathcal{X}_j|+4$ for $j=1,2$, and $|\mathcal{Q}|\leq 6$.
  For a fix $P_1^*(\cdot)\in \mathcal{P}_1^*$, let $\mathcal{R}_{HK}(P^*)$ be the set of $(R_1,R_2)$ satisfying
  \begin{align}
  0\leq R_1 &\leq I(X_1;Y_1|W_2Q)\triangleq a, \label{HK-1}\\
  0\leq R_2 &\leq I(X_2;Y_2|W_1Q)\triangleq b, \label{HK-2}\\
  R_1+R_2 &\leq I(X_1W_2;Y_1|Q)+I(X_2;Y_2|W_1W_2Q)\triangleq c, \label{HK-3}\\
  R_1+R_2 &\leq I(X_1;Y_1|W_1W_2Q)+I(X_2W_1;Y_2|Q)\triangleq d, \label{HK-4}\\
  R_1+R_2 &\leq I(X_1W_2;Y_1|W_1Q)+I(X_2W_1;Y_2|W_2Q)\triangleq e, \label{HK-5}\\
  2R_1+R_2 &\leq I(X_1W_2;Y_1|Q)+I(X_1;Y_1|W_1W_2Q)+I(X_2W_1;Y_2|W_2Q)\triangleq f, \label{HK-6}\\
  R_1+2R_2 &\leq I(X_2;Y_2|W_1W_2Q)+I(X_2W_1;Y_2|Q)+I(X_1W_2;Y_1|W_1Q)\triangleq g. \label{HK-7}
  \end{align}
  Then we have that
  \begin{equation}
  \mathcal{R}_{HK}=\bigcup_{P_1^*\in \mathcal{P}_1^*}\mathcal{R}_{HK}(P_1^*)
  \end{equation}
  is an achievable rate region for the DM-IC.
  \end{theorem}

  It is shown in \cite{chong2008han} that the regions described in Theorem \ref{theorem:HK-o} and \ref{theorem:HK} are equivalent, and constraints (\ref{HK1-2}), (\ref{HK1-3}) and (\ref{HK1-6}) and their counterparts for the second receiver are unnecessary. It is straightforward to see that $\mathcal{R}_{HK}^o(P^*)\subseteq \mathcal{R}_{HK}(P_1^*)$ by using Fourier-Motzkin elimination \cite{chong2008han}, where $\mathcal{R}_{HK}^o(P^*)=\mathcal{R}_{HK}^{o,1}(P^*)\cap \mathcal{R}_{HK}^{o,2}(P^*)$ and $$P_1^*(q,w_1.w_2,x_1,x_2)=\sum_{v_1\in\mathcal{V}_1,v_2\in\mathcal{V}_2}P^*(q,v_1,v_2,w_1.w_2,x_1,x_2).$$ However, to prove the converse, we will need \cite[Lemma 2]{chong2008han}, which states that $\mathcal{R}_{HK}(P_1^*)\subseteq \mathcal{R}_{HK}^o(P^*)\cup\mathcal{R}_{HK}^o(P^{**})\cup\mathcal{R}_{HK}^o(P^{***})$, where
  \begin{align*}
  P^{**}=\sum_{w_1\in\mathcal{W}_1}P^*,~~~~
  P^{***}=\sum_{w_2\in\mathcal{W}_2}P^*.
  \end{align*}
  This indicates that for a given joint distribution, the original Han-Kobayashi region can be smaller than the compact one, as shown in \cite[Remark 3]{chong2008han}. Thus, to achieve $\mathcal{R}_{HK}(P_1^*)$ for a some $P_1^*$ with the scheme of \cite{wang2015channel}, one generally will need to use three codes designed for different joint distributions. In this paper, we aim to design a heterogeneous superposition polar coding scheme to achieve $\mathcal{R}_{HK}(P_1^*)$ directly.

  \subsection{Partially-Joint Decoding for the 2-User DM-IC}
  \label{S:PJD}
  To achieve the whole Han-Kobayashi region, both superposition coding variants require joint decoding of all component messages at each receiver, which we refer to as fully-joint decoding. For the homogeneous variant, fully-joint decoding can be realized by polar codes using MAC polarization techniques since each component message is independently encoded, as \cite{wang2015channel} has adopted. For the heterogeneous variant, however, fully-joint decoding may not be easily implemented using polar codes and practical decoding algorithms (such as SCD), as the coarse message and the fine message are encoded sequentially. When decoding the fine message in a heterogeneous superposition polar coding scheme (such as \cite{goela2015broadcast}), the estimate of the coarse message is required as side information. To design a polar coding scheme with practical decoding algorithm that can achieve $\mathcal{R}_{HK}(P_1^*)$ directly, we propose and prove two types of partially-joint decoding orders.
  \begin{definition}[Partially-joint decoding]
  	The two types of partially-joint decoding are defined as:
  	\begin{itemize}
  		\item (Type I) a receiver jointly decodes two senders' common messages first, and then decodes its private message with the estimates of the common messages;
  		\item (Type II) a receiver decodes its intended common message first, and then jointly decodes the unintended common message and its private message with the estimate of the intended common message.
  	\end{itemize}
  \end{definition}

  \begin{theorem}
  	\label{theorem:HK-partial}
  	Let $\mathcal{R}_{Par}^1(P_1^*)$ be the achievable rate region of the DM-IC when both receivers use the Type I partially-joint decoding, and $\mathcal{R}_{Par}^2(P_1^*)$ (resp. $\mathcal{R}_{Par}^3(P_1^*)$) the region when Receiver 1 (resp. 2) adopts Type I while Receiver 2 (resp.1) applies Type II. Define $\mathcal{R}_{Par}(P_1^*)=\mathcal{R}_{Par}^1(P_1^*)\cup \mathcal{R}_{Par}^2(P_1^*)\cup \mathcal{R}_{Par}^3(P_1^*)$. Then we have
  	\begin{equation}
  	\mathcal{R}_{Par}(P_1^*)=\mathcal{R}_{HK}(P_1^*).
  	\end{equation}
  \end{theorem}
  \begin{proof}
  	See Appendix \ref{APPEN-A}.
  \end{proof}
  
  \newtheorem{remark}{Remark}
  \begin{remark}
  	It is worth noting that we do not consider the case when both receivers use the Type II partially-joint decoding. This is because the Han-Kobayashi region can already be covered by the other three decoding strategies. In fact, one can easily verify that the achievable rate region in this case can also be achieved by at least one of the other three strategies since the upper bounds on the common message rates ($R_k^c\leq I(W_k;Y_k|Q)$ for $k=1,2$) are non-optimal. This explains why in our proposed polar coding scheme in Section \ref{S:Overview} we do not need such an approach either.
  \end{remark}
  
  \begin{remark}
  	\label{Remark:2}
  	 The reasons why the fully-joint decoder is hard to design are twofold, the decoding algorithm and the code structure.  	 
  	 Existing polar codes are optimized for SCD, which is sequential in nature. To design a joint decoding scheme using SCD, one has to use methods similar to the permutation based MAC polarization -- mixing different users' sequences of random variables into a single one and then decoding them together. However, in the heterogeneous scheme, $W_k$ and $X_k$ ($k=1,2$) are correlated. If we try to apply this method, the induced random process will have a complicated memory. Although there have been studies on polarization for processes with memory \cite{sasoglu2011polar,wang2015memory,sasoglu2016memory}, the results are still far from handling such a problem now.
  	 If we want to realize genuine fully-joint decoding (e.g., using maximum-likelihood (ML) or ML-like decoding), then the corresponding structure of codes should also be optimized for this decoding algorithm (we cannot use the same code structure optimized for SCD and just switch to ML decoding, as the achievable rate region of the scheme remains the same). However, neither the construction complexity nor the decoding complexity is affordable. 
  \end{remark}

\section{Polar Coding Preliminaries}
\label{S:PolarPri}
  
  \subsection{Polar Coding for Lossless Source Compression}
  First, let us recap the lossless source polarization scheme introduced in \cite{arikan2010source} and generalized to arbitrary alphabet in \cite{sasoglu2011polar}. Let $(X,Y)\sim p_{X,Y}$ be a pair of random variables over $(\mathcal{X}\times \mathcal{Y})$ with $|\mathcal{X}|=q_X$ being a prime number\footnote{Although for composite $q_X$, polarization can also happen if we use some special types of operations instead of group operation \cite{sasoglu2011polar,Nasser2016arbitrary,Nasser2016I,Nasser2017II}, we only consider the prime number case in this paper for simplicity.}. Consider $X$ as the memoryless source to be compressed and $Y$ as \textit{side information} of $X$. Let $U^{1:N}=X^{1:N}\mathbf{G}_N$. As $N$ goes to infinity, $U^j$ ($j\in [N]$) becomes either almost independent of $(Y^{1:N},U^{1:j-1})$ and uniformly distributed, or almost determined by $(Y^{1:N},U^{1:j-1})$ \cite{arikan2010source}. Define the following sets of polarized indices:
  \begin{align}
  \mathcal{H}^{(N)}_{X|Y}&=\{j\in [N]:H(U^j|Y^{1:N},U^{1:j-1})\geq \log_2(q_X)-\delta_N\},\label{HXY}\\
  \mathcal{L}^{(N)}_{X|Y}&=\{j\in [N]:H(U^j|Y^{1:N},U^{1:j-1})\leq \delta_N\}.\label{LXY}
  \end{align}
  From \cite{sasoglu2011polar,chou2016broadcast} we have
  \begin{equation}
  \label{PolarRate}
  \begin{aligned}
  \lim_{N\rightarrow \infty}\frac{1}{N}|\mathcal{H}^{(N)}_{X|Y}|&=H_{q_X}(X|Y),\\
  \lim_{N\rightarrow \infty}\frac{1}{N}|\mathcal{L}^{(N)}_{X|Y}|&=1-H_{q_X}(X|Y).
  \end{aligned}
  \end{equation}
  With $U^{(\mathcal{L}^{(N)}_{X|Y})^C}$ and $Y^{1:N}$, $X^{1:N}$ can be recovered at arbitrarily low error probability given sufficiently large $N$.
  
  The compression of a single source $X$ can be seen as a special case of the above one by letting $Y=\emptyset$.

\subsection{Polar Coding for Arbitrary Discrete Memoryless Channels}
  Polar codes were originally developed for symmetric channels. By invoking results in source polarization, one can construct polar codes for asymmetric channels without alphabet extension, as introduced in \cite{honda2013asymmetric}. However, the scheme of \cite{honda2013asymmetric} requires the encoder and the decoder to share a large amount of random mappings, which raises a practical concern of not being explicit. In \cite{goela2015broadcast,chou2015deterministic,chou2016broadcast,gad2016asymmetric}, deterministic mappings are used to replace (part of) the random mappings so as to reduce the amount of shared randomness needed. Next, we briefly review the method of \cite{chou2015deterministic,chou2016broadcast}\footnote{We note that the common message encoding scheme in \cite{chou2016broadcast} (consider the special case when there is no Eve and no chaining scheme) and the scheme in \cite{chou2015deterministic} share the same essence, although there is a slight difference in the partition scheme for information and frozen symbols (see (11) of \cite{chou2016broadcast} and (10) of \cite{chou2015deterministic}), and reference \cite{chou2015deterministic} uses deterministic rules for some symbols while reference \cite{chou2016broadcast} uses random rules.}, which only requires a vanishing rate of shared randomness.
  
  Let $W(Y|X)$ be a discrete memoryless channel (DMC) with a $q_X$-ary input alphabet $\mathcal{X}$, where $q_X$ is a prime number. Let $U^{1:N}=X^{1:N}\mathbf{G}_N$ and define $\mathcal{H}^{(N)}_X$ and $\mathcal{H}^{(N)}_{X|Y}$ as in (\ref{HXY}), and $\mathcal{L}^{(N)}_{X|Y}$ as in (\ref{LXY}). Define the information set, frozen set and almost deterministic set respectively as follows:
  \begin{align}
  \mathcal{I}&\triangleq \mathcal{H}_X^{(N)}\cap \mathcal{L}_{X|Y}^{(N)}, \label{PCAC-I}\\
  \mathcal{F}_r&\triangleq \mathcal{H}_X^{(N)}\cap (\mathcal{L}_{X|Y}^{(N)})^c, \label{PCAC-Fr}\\
  \mathcal{F}_d&\triangleq (\mathcal{H}_X^{(N)})^c.\label{PCAC-Fd}
  \end{align}
  The encoding procedure goes as follows: $\{u^j\}_{j\in \mathcal{I}}$ carry information, $\{u^j\}_{j\in \mathcal{F}_r}$ are filled with uniformly distributed frozen symbols (shared between the sender and the receiver), and $\{u^j\}_{j\in \mathcal{F}_d}$ are randomly generated according to conditional probability $P_{U^j|U^{1:j-1}}(u|u^{1:j-1})$. To guarantee reliable decoding, $\{u^j\}_{j\in (\mathcal{H}_X^{(N)})^C\cap (\mathcal{L}_{X|Y}^{(N)})^C}$ are separately transmitted to the receiver with some reliable error-correcting code, the rate of which vanishes as $N$ goes large \cite{chou2016broadcast}. Since $\{u^j\}_{j\in \mathcal{F}_r}$ only need to be uniformly distributed, they can be the same in different blocks. Thus, the rate of frozen symbols in this scheme can also be made negligible by reusing them over sufficient number of blocks.
   
  After receiving $y^{1:N}$ and recovered $\{u^j\}_{j\in (\mathcal{H}_X^{(N)})^C\cap (\mathcal{L}_{X|Y}^{(N)})^C}$, the receiver computes the estimate $\bar{u}^{1:N}$ of $u^{1:N}$ with a SCD as
  \begin{equation}
  \bar{u}^{j}=
  \begin{cases}
  u^j,&\text{if } j\in (\mathcal{L}_{X|Y}^{(N)})^C\\
  \arg\max_{u\in\{0,1\}}P_{U^{j}|Y^{1:N}U^{1:j-1}}(u|y^{1:N},u^{1:j-1}),&\text{if } j\in \mathcal{L}_{X|Y}^{(N)}
  \end{cases}.
  \end{equation}
  
  It is shown that the rate of this scheme, $R=|\mathcal{I}|/N$, satisfies \cite{honda2013asymmetric}
  \begin{equation}
  \lim_{N\rightarrow \infty}R=I(X;Y).
  \end{equation}

  \subsection{Polar Coding for Multiple Access Channels}
  \label{Sec-MAC}
  Let $P_{Y|X_1X_2}(y|x_1,x_2)$ be the transition probability of a discrete memoryless 2-user MAC, where $x_1\in\mathcal{X}_1$ with $|\mathcal{X}_1|=q_{X_1}$ and $x_2\in\mathcal{X}_2$ with $|\mathcal{X}_2|=q_{X_2}$. For a fixed product distribution of $P_{X_1}(x_1)P_{X_2}(x_2)$, the achievable rate region of $P_{Y|X_1X_2}$ is given by \cite{cover2012informtaion}
  \begin{equation}
  \mathcal{R}(P_{Y|X_1X_2}) \triangleq \left\lbrace
  \begin{matrix}
  \left(
  \begin{array}{ccc}
  R_1\\
  R_2
  \end{array}
  \right)&
  \left|
  \begin{array}{ccc}
  \begin{array}{ccc}
  0 \leq R_1 \leq I(X_1;Y|X_2)\\
  0 \leq R_2 \leq I(X_2;Y|X_1)\\
  R_1+R_2 \leq I(X_1,X_2;Y)
  \end{array}
  \end{array}\right.
  \end{matrix}
  \right\rbrace .
  \end{equation}
  
  Polar coding for MACs has been studied in \cite{arikan2012sw,onayscmac,sasoglu2013mac,abbe2012mmac,mahdavifar2016uniform}. Although \cite{mahdavifar2016uniform} provides a more general scheme that can achieve the whole uniform rate region of a $m$-user ($m\geq 2$) MAC, in our scheme, we adopt the monotone chain rule expansion method in \cite{arikan2012sw} because it has simple structure and possesses similar complexity to the single-user polar codes. Reference \cite{arikan2012sw} mainly deals with the Slepian-Wolf problem in source coding, but the method can be readily applied to the problem of coding for the 2-user MAC since they are dual problems, which has been studied in \cite{onayscmac} and used in \cite{wang2015channel}. However, both \cite{onayscmac} and \cite{wang2015channel} consider uniform channel inputs. Here we generalize it to arbitrary input case with the approach of the previous subsection. Note that although the input alphabets of the two users can be different, the extension is straightforward since there is no polarization operation between the two channel inputs. 
  For simplicity, we assume $q_{X_1}$ and $q_{X_2}$ are prime numbers.
  Define
  \begin{equation}
  \label{UX}
  U_1^{1:N}=X_1^{1:N}\mathbf{G}_N,~~U_2^{1:N}=X_2^{1:N}\mathbf{G}_N.
  \end{equation}
  Let $S^{1:2N}$ be a permutation of $U_1^{1:N}U_2^{1:N}$ such that it preserves the relative order of the elements of both $U_1^{1:N}$ and $U_2^{1:N}$, called a monotone chain rule expansion. Such an expansion can be represented by a string $\mathbf{b}_{2N}=b_1b_2...b_{2N}$, called the \textit{path} of the expansion, where $b_j=0$ $(j\in [2N])$ represents that $S^j\in U_1^{1:N}$, and $b_j=1$ represents that $S^j\in U_2^{1:N}$. Then we have
  \begin{align*}
  I(Y^{1:N};U_1^{1:N},U_2^{1:N})&=H(U_1^{1:N},U_2^{1:N})-H(U_1^{1:N},U_2^{1:N}|Y^{1:N})\\
  &=NH(X_1)+NH(X_2)-\sum_{j=1}^{2N}H(S^j|Y^{1:N},S^{1:j-1}).
  \end{align*}
  It is shown in \cite{arikan2012sw} that $H(S^j|Y^{1:N},S^{1:j-1})$ ($j\in [2N]$) polarizes to 0 or 1 as $N$ goes to infinity\footnote{The entropy here is calculated adaptively. If $j\in\mathcal{S}_{U_k}$ ($k=1,2$), then entropy is calculated with $q_{X_k}$-based logarithm.}. Define the rates of the two users as
  \begin{equation}
  \begin{aligned}
  R_{U_1}&=H(X_1)-\frac{1}{N}\sum_{j\in \mathcal{S}_{U_1}}H(S^j|Y^{1:N},S^{1:j-1}),\\
  R_{U_2}&=H(X_2)-\frac{1}{N}\sum_{j\in \mathcal{S}_{U_2}}H(S^j|Y^{1:N},S^{1:j-1}),
  \end{aligned}
  \end{equation}
  respectively, where $\mathcal{S}_{U_1}\triangleq \{ j\in [2N]:b_j=0\}$ and $\mathcal{S}_{U_2}\triangleq \{ j\in [2N]:b_j=1\}$.
  \newtheorem{proposition}{Proposition}
  \begin{proposition}[\cite{arikan2012sw}]
  	\label{proposition:MAC-1}
  	Let $(R_1,R_2)$ be a rate pair on the dominant face of $\mathcal{R}(P_{Y|X_1X_2})$. For any given $\epsilon >0$, there exists $N$ and a chain rule $\mathbf{b}_{2N}$ on $U_1^{1:N}U_2^{1:N}$ such that $\mathbf{b}_{2N}$ is of the form $0^i1^N0^{N-i}$ ($0\leq i\leq N$) and has a rate pair $(R_{U_1},R_{U_2})$ satisfying
  	\begin{equation}
  	|R_1-R_{U_1}|\leq \epsilon \text{ and } |R_2-R_{U_2}|\leq \epsilon.
  	\end{equation}
  \end{proposition}
  
  Although the permutations can have lots of variants, even non-monotone \cite{mahdavifar2016uniform}, Proposition \ref{proposition:MAC-1} shows that expansions of type $0^i1^N0^{N-i}$ ($0\leq i\leq N$) are sufficient to achieve every point on the dominant face of $\mathcal{R}(P_{Y|X_1X_2})$ given sufficiently large $N$, which can make our code design and construction simpler. To polarize a MAC sufficiently while keeping the above rate approximation intact, we need to scale the path. For any integer $l=2^m$, let $l\mathbf{b}_{2N}$ denote $$\underbrace{b_1\cdots b_1}_l \underbrace{b_2\cdots b_2}_l\cdots\cdots \underbrace{b_{2N}\cdots b_{2N}}_l,$$ which is a monotone chain rule for $U_1^{1:lN}U_2^{1:lN}$. It is shown in \cite{arikan2012sw} that the rate pair for $\mathbf{b}_{2N}$ is also the rate pair for $l\mathbf{b}_{2N}$.

  Now we can construct a polar code for the 2-user MAC with arbitrary inputs.  Let $f_k(i):[N]\rightarrow \mathcal{S}_{U_k}$ ($k=1,2$) be the mapping from indices of $U_k^{1:N}$ to those of $S^{\mathcal{S}_{U_k}}$. Define
  \begin{equation}
  \begin{aligned}
  \mathcal{H}^{(N)}_{S_{U_k}}&\triangleq \{j\in [N]:H(S^{f_k(j)}|S^{1:f_k(j)-1})\geq \log_2({q_{X_k}})-\delta_N\},\\
  \mathcal{L}^{(N)}_{S_{U_k}|Y}&\triangleq \{j\in [N]:H(S^{f_k(j)}|Y^{1:N},S^{1:f_k(j)-1})\leq \delta_N\},
  \end{aligned}
  \end{equation}
  which satisfy
  \begin{equation}
  \label{MACRate}
  \begin{aligned}
  \lim_{N\rightarrow \infty}\frac{1}{N}|\mathcal{H}^{(N)}_{S_{U_k}}|&=\frac{1}{N}\sum_{j\in \mathcal{S}_{U_k}}H_{q_{X_k}}(S^j|Y^{1:N},S^{1:j-1}),\\
  \lim_{N\rightarrow \infty}\frac{1}{N}|\mathcal{L}^{(N)}_{S_{U_k}}|&=1-\frac{1}{N}\sum_{j\in \mathcal{S}_{U_k}}H_{q_{X_k}}(S^j|Y^{1:N},S^{1:j-1}).
  \end{aligned}
  \end{equation}
  Since $X_1$ and $X_2$ are independent, we have
  \begin{equation}
  \mathcal{H}^{(N)}_{S_{U_k}}=\mathcal{H}^{(N)}_{X_k}\triangleq \{j\in [N] :H(U_k^j|U_k^{1:j-1})\geq \log_2({q_{X_k}})-\delta_N\}.
  \end{equation}
  Partition user $k$'s ($k=1,2$) indices as
  \begin{equation}
  \begin{aligned}
  \mathcal{I}_k&\triangleq \mathcal{H}^{(N)}_{S_{U_k}}\cap \mathcal{L}^{(N)}_{{S_{U_k}}|Y}, \\
  \mathcal{F}_{kr}&\triangleq \mathcal{H}^{(N)}_{S_{U_k}}\cap (\mathcal{L}^{(N)}_{{S_{U_k}}|Y})^C, \\
  \mathcal{F}_{kd}&\triangleq (\mathcal{H}^{(N)}_{S_{U_k}})^C.
  \end{aligned}
  \end{equation}
 Then each user can apply the same encoding scheme as the single-user case. The receiver uses a SCD to decode two users' information jointly according to the expansion order. The polarization result can be summarized as the following proposition.
  \begin{proposition}[\cite{arikan2012sw}]
  	\label{proposition:MAC-2}
  	Let $P_{Y|X_1X_2}(y|x_1,x_2)$ be the transition probability of a discrete memoryless 2-user MAC. Consider the transformation defined in (\ref{UX}). Let $N_0=2^{n_0}$ for some $n_0\geq 1$ and fix a path $\mathbf{b}_{2N_0}$ for $U_1^{1:N_0}U_2^{1:N_0}$. The rate pair for $\mathbf{b}_{2N_0}$ is denoted by $(R_{U_1},R_{U_2})$. Let $N=2^lN_0$ for $l\geq 1$ and let $S^{1:2N}$ be the expansion represented by $2^l\mathbf{b}_{2N_0}$. Then, for any given $\delta>0$, as $l$ goes to infinity, we have (the entropy here is also calculated adaptively)
  	\begin{equation}
  	\label{MACR}
  	\begin{aligned}
  	&\frac{1}{2N}\big{|}\{ 1\leq j\leq 2N:\delta <H(S^j|Y^{1:N},S^{1:j-1})<1-\delta \}\big{|}\rightarrow 0,\\
  	&~~~~~~\frac{|\mathcal{I}_1|}{N}\rightarrow R_{U_1} \text{ and } \frac{|\mathcal{I}_2|}{N}\rightarrow R_{U_2}.
  	\end{aligned}
  	\end{equation}
  \end{proposition}
  
  Proposition \ref{proposition:MAC-1} and \ref{proposition:MAC-2} can be readily extended from Theorem 1 and Theorem 2 in \cite{arikan2012sw} by considering $Y$ as side information of source pair $(X_1,X_2)$ and performing the same analysis. Thus, we omit the proof here.
   
\section{An Overview of Our New Approach}
\label{S:Overview}
\begin{figure}[tb]
	\centering
	\includegraphics[width=11cm]{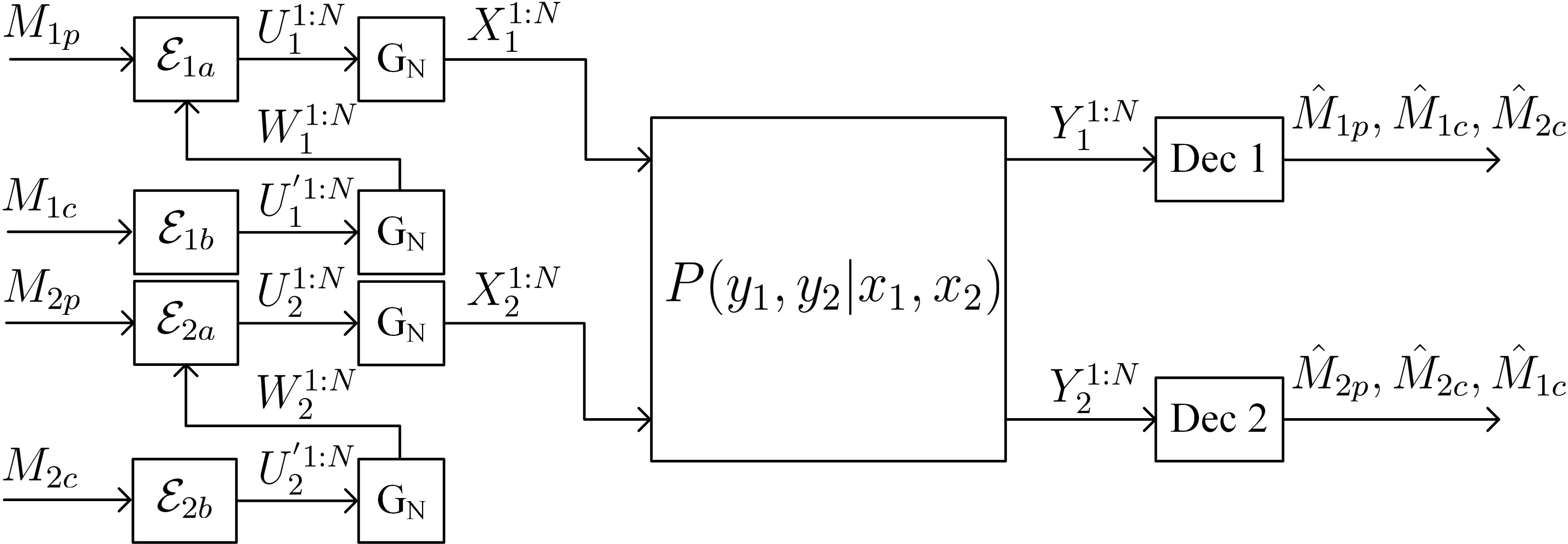}
	\caption{Proposed heterogeneous superposition polar coding scheme for the 2-user DM-IC.} \label{fig:ic-hk}
\end{figure}

  In this section, we introduce the main idea of our scheme. Since the purpose of introducing the time-sharing parameter $Q$ in Theorem \ref{theorem:HK-o} and \ref{theorem:HK} is to replace the convex-hull operation, in the code design part, we will consider a fixed $Q=q$ and drop this condition in the expressions for simplicity. 
  
  Our proposed heterogeneous superposition polar coding scheme is illustrated in Fig. \ref{fig:ic-hk}. Sender $k$'s ($k=1,2$) splits its message $M_k$ into a private message $M_{kp}$ and a common message $M_{kc}$. Encoder $\mathcal{E}_{kb}$ maps $M_{kc}$ into a sequence $U_k^{'1:N}$ of length $N$, which goes into a polar encoder to generate an intermediate codeword $W_k^{1:N}$ (corresponding to ARV $W_k$ in Theorem \ref{theorem:HK}). Encoder $\mathcal{E}_{ka}$ then maps $M_{kp}$ together with $W_k^{1:N}$ into $U_k^{1:N}$, which goes into another polar encoder to generate the final codeword $X_k^{1:N}$.

  \subsection{Synthesized MACs for Receivers}
    For a target rate pair $\mathbf{P}$, let $R_k^p$ and $R_k^{c}$ respectively denote the corresponding private and common message rates of Sender $k$ ($k=1,2$), and define $\mathbf{P}^1\triangleq (R_1^p+R_1^c,R_2^c)$ and $\mathbf{P}^2\triangleq (R_1^c,R_2^p+R_2^c)$ as Receiver 1's and Receiver 2's receiving rate pairs respectively. Furthermore, define $\mathbf{P}^c\triangleq (R_1^c,R_2^c)$ as the common message rate pair. In the rest of this paper, we refer $(R_1^p,R_1^c,R_2^p,R_2^c)$ to a \textit{rate decomposition} of $\mathbf{P}$.
    
    For the purpose of decomposing a target rate pair into a private and common message rate tuple suitable for our partially-joint decoding scheme, we first define the \textit{effective channel} of each receiver. For Receiver 1, its effective channel, $P_{Y_1|X_1W_2}$, is defined as
    \begin{equation}
    P_{Y_1|X_1W_2}(y_1|x_1,w_2)\triangleq \sum_{x_2}{P_{Y_1|X_1X_2}(y_1|x_1,x_2)P_{X_2|W_2Q}(x_2|w_2,q)}.
    \end{equation}
    Similarly, the effective channel of Receiver 2 is defined as
    \begin{equation}
    P_{Y_2|W_1X_2}(y_2|w_1,x_2)\triangleq \sum_{x_1}{P_{Y_2|X_1X_2}(y_2|x_1,x_2)P_{X_1|W_1Q}(x_1|w_1,q)}.
    \end{equation}
    The achievable rate regions for these two MACs are
    \begin{equation}
    \mathcal{R}(P_{Y_1|X_1W_2}) = \left\lbrace
    \begin{matrix}
    \left(
    \begin{array}{ccc}
    R_1\\
    R_2
    \end{array}
    \right)&
    \left|
    \begin{array}{ccc}
    \begin{array}{ccc}
    0 \leq R_1 \leq I(X_1;Y_1|W_2)\\
    0 \leq R_2 \leq I(W_2;Y_1|X_1)\\
    R_1+R_2 \leq I(X_1W_2;Y_1)
    \end{array}
    \end{array}\right.
    \end{matrix}
    \right\rbrace ,
    \end{equation}
    \begin{equation}
    \mathcal{R}(P_{Y_2|W_1X_2}) = \left\lbrace
    \begin{matrix}
    \left(
    \begin{array}{ccc}
    R_1\\
    R_2
    \end{array}
    \right)&
    \left|
    \begin{array}{ccc}
    \begin{array}{ccc}
    0 \leq R_1 \leq I(W_1;Y_2|X_2)\\
    0 \leq R_2 \leq I(X_2;Y_2|W_1)\\
    R_1+R_2 \leq I(X_2W_1;Y_2)
    \end{array}
    \end{array}\right.
    \end{matrix}
    \right\rbrace.
    \end{equation}
    
    Now we can study the Han-Kobayashi coding problem in $P_{Y_1|X_1W_2}$ and $P_{Y_2|W_1X_2}$. In these two MACs, the rate of $X_k$ ($k=1,2$) equals the overall rate of Sender $k$, while the rate of $W_k$ equals the common message rate of Sender $k$. Obviously, $\mathbf{P}^1$ and $\mathbf{P}^2$ must lie inside $\mathcal{R}(P_{Y_1|X_1W_2})$ and $\mathcal{R}(P_{Y_2|W_1X_2})$ respectively in order to make reliable communication possible.
    
    Giving only two effective channels is insufficient to determine the suitable decoding order for a target rate pair. If we hope to use a partially-joint decoder, the following two MACs, $P_{Y_1|W_1W_2}$ and $P_{Y_2|W_1W_2}$, will be useful. For $k=1,2$, define
    \begin{equation}
    P_{Y_k|W_1W_2}(y_k|w_1,w_2)\triangleq \sum_{x_1}\sum_{x_2}{P_{Y_k|X_1X_2}(y_k|x_1,x_2)P_{X_1|W_1Q}(x_1|w_1,q)P_{X_2|W_2Q}(x_2|w_2,q)},
    \end{equation}
    the achievable rate region of which is
    \begin{equation}
    \mathcal{R}(P_{Y_k|W_1W_2}) = \left\lbrace
    \begin{matrix}
    \left(
    \begin{array}{ccc}
    R_1\\
    R_2
    \end{array}
    \right)&
    \left|
    \begin{array}{ccc}
    \begin{array}{ccc}
    0 \leq R_1 \leq I(W_1;Y_k|W_2)\\
    0 \leq R_2 \leq I(W_2;Y_k|W_1)\\
    R_1+R_2 \leq I(W_1W_2;Y_k)
    \end{array}
    \end{array}\right.
    \end{matrix}
    \right\rbrace .
    \end{equation}
    \begin{figure}[tb]
    	\centering
    	\includegraphics[width=11cm]{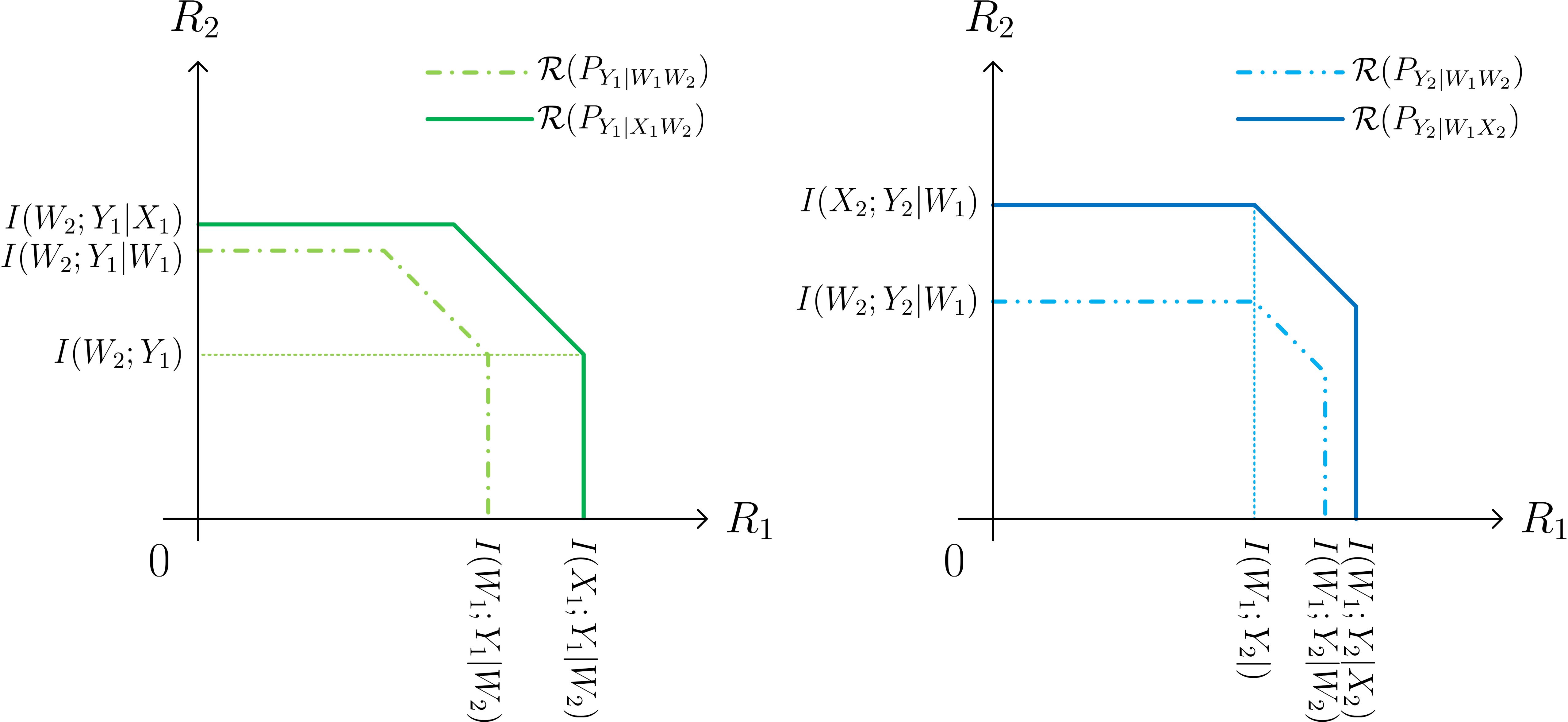}
    	\caption{Illustration for the achievable rate regions of the synthesized MACs.} \label{fig:mac_relation}
    \end{figure}
    
    The relations between the above four achievable rate regions are shown in Fig. \ref{fig:mac_relation}. If the common message rate pair lies inside $\mathcal{R}(P_{Y_k|W_1W_2})$, then Receiver $k$ can apply the Type I partially-joint decoding. Otherwise it will need to use the Type II one.

  \subsection{The General Idea of Our Scheme}
   According to the two receivers' different choices of partially-joint decoding orders, we define the following two types of points (rate pairs).   
   \begin{definition}[Type A points]
   	A Type A point $\mathbf{P}$ in $\mathcal{R}_{HK}(P_1^*)$ is a rate pair which can be decomposed into a private and common message rate tuple that satisfies:
   	\begin{equation}
   	\begin{aligned}
   	(R_1^c,R_2^c)&\in \mathcal{R}(P_{Y_1|W_1W_2})\cap\mathcal{R}(P_{Y_2|W_1W_2}),\\
   	R_1^p&=I(X_1;Y_1|W_1W_2),\\
   	R_2^p&=I(X_2;Y_2|W_1W_2).
   	\end{aligned}
   	\end{equation} 
   \end{definition}
   
   \begin{definition}[Type B points]
   	\label{def:TYII}
   	A Type B point $\mathbf{P}$ in $\mathcal{R}_{HK}(P_1^*)$ is a rate pair which can be decomposed into a private and common message rate tuple that satisfies:
   	\begin{equation}
   	\begin{aligned}
   	(R_1^c,R_2^c)&\in \mathcal{R}(P_{Y_k|W_1W_2}),\\
   	R_{k'}^c&\leq I(W_{k'};Y_{k'}),\\
   	R_k^p&=I(X_k;Y_k|W_1W_2),\\
   	R_{k'}^p&=I(X_{k'}W_k;Y_{k'}|W_{k'})-R_k^c,
   	\end{aligned} 
   	\end{equation} 
   	where $k,k'\in\{1,2\}$ and $k\neq k'$.
   \end{definition}
   
  To achieve a Type A point $\mathbf{P}$, both receivers can apply the Type I partially-joint decoding. We first design a MAC polar code for two common messages that achieves $\mathbf{P}^c$ in the compound MAC composed of $P_{Y_1|W_1W_2}$ and $P_{Y_2|W_1W_2}$, and then design a point-to-point polar code for each sender's private message with the common messages being side information.   
  To achieve a Type B point, one receiver applies the Type I partially-joint decoding while the other applies Type II. Let us consider $k=2,k'=1$ as an example. The code structures for two common messages $(M_{1c},M_{2c})$ and Sender 1's private message $M_{1p}$ are jointly designed in such a way that, Receiver 1 can first decode $M_{1c}$ (equivalently $W_1^{1:N}$) with $Y_1^{1:N}$ and then jointly decode $(M_{1p},M_{2c})$ with the estimate of $W_1^{1:N}$, while Receiver 2 can jointly decode $(M_{1c},M_{2c})$ with $Y_2^{1:N}$. The code structure for Sender 2's private message $M_{2p}$ is simply point-to-point polar codes.
  
  In Section \ref{S:PJD} we have proved by random coding that partially-joint decoding can achieve the whole Han-Kobayashi region. The following lemma provides another evidence to support this conclusion.
  \newtheorem{lemma}{Lemma}
  \begin{lemma}
  	\label{lemma:typeAB}
  Every point on the dominant faces of $\mathcal{R}_{HK}(P_1^*)$ can be classified into either Type A or Type B.
  \end{lemma}
  \begin{proof}
	See Appendix \ref{APPEN-B}.
  \end{proof}

  \section{Proposed Polar Coding Schemes}
  \label{S:PolarScheme}
  
  In this section, we describe details of our proposed two types of polar coding schemes for the 2-user DM-IC. We consider the case when $q_{X_1}=|\mathcal{X}_1|$ and $q_{X_2}=|\mathcal{X}_2|$ are two prime numbers, $q_{W_1}=|\mathcal{W}_1|$ is the smallest prime number larger than $q_{X_1}+4$, and $q_{W_2}=|\mathcal{W}_2|$ is the smallest prime number larger than $q_{X_2}+4$. For a rate pair $\mathbf{P}$, let $\mathbf{P}(1)$ and $\mathbf{P}(2)$ respectively denote its first and second component.
  
  \subsection{Common Message Encoding}
  \label{S:CommEnc}
  \subsubsection{Partition Scheme for Type A Points}
  
  Let $\mathbf{P}^c=(R_1^c,R_2^c)$ be the common message rate pair for a Type A point $\mathbf{P}$ on a dominant face of $\mathcal{R}_{HK}(P_1^*)$. Obviously, $\mathbf{P}^c$ must lie on the dominant face of either $\mathcal{R}(P_{Y_1|W_1W_2})$ or $\mathcal{R}(P_{Y_2|W_1W_2})$, otherwise we can choose a larger common message rate pair to achieve higher rates. Without loss of generality, we assume that $\mathbf{P}^c$ is on the dominant face of $\mathcal{R}(P_{Y_1|W_1W_2})$ in this subsection as an example.
  
  First, choose a point $\mathbf{\tilde{P}}^c$ on the dominant face of $\mathcal{R}(P_{Y_2|W_1W_2})$ which is larger than $\mathbf{P}^c$ in the sense that $\mathbf{\tilde{P}}^c(1)\geq \mathbf{P}^c(1)$ and $\mathbf{\tilde{P}}^c(2)\geq \mathbf{P}^c(2)$, as the target point for conducting the monotone chain rule expansion in our code design. Let $S^{1:2N}$ be the monotone chain rule expansion that achieves $\mathbf{P}^c$ in $\mathcal{R}(P_{Y_1|W_1W_2})$, and $T^{1:2N}$ the expansion that achieves $\mathbf{\tilde{P}}^c$ in $\mathcal{R}(P_{Y_2|W_1W_2})$. Denote the sets of indices in $S^{1:2N}$ with $S^j\in U_1^{'1:N}$ and $S^j\in U_2^{'1:N}$ by $\mathcal{S}_{U'_1}$ and $\mathcal{S}_{U'_2}$ respectively, and those in $T^{1:2N}$ with $T^j\in U_1^{'1:N}$ and $T^j\in U_2^{'1:N}$ by $\mathcal{T}_{U'_1}$ and $\mathcal{T}_{U'_2}$ respectively. For $k=1,2$, let $f_k(j):[N]\rightarrow \mathcal{S}_{U'_k}$ be the mapping from indices of $U_k^{'1:N}$ to those of $S^{\mathcal{S}_{U'_k}}$, and $g_k(j):[N]\rightarrow \mathcal{T}_{U'_j}$ the mapping from indices of $U_k^{'1:N}$ to those of $T^{\mathcal{T}_{U'_k}}$. Define the following polarized sets
  \begin{equation}
  \begin{aligned}
  \mathcal{H}^{(N)}_{S_{U'_k}}&\triangleq \big{\{}j\in [N]:H(S^{f_k(j)}|S^{1:f_k(j)-1})\geq \log_2({q_{W_k}})-\delta_N \big{\}},\\
  \mathcal{L}^{(N)}_{S_{U'_k}|Y_1}&\triangleq \big{\{}j\in [N]:H(S^{f_k(j)}|Y_1^{1:N}, S^{1:f_k(j)-1})\leq \delta_N \big{\}},\\
  \mathcal{H}^{(N)}_{T_{U'_k}}&\triangleq \big{\{}j\in [N]:H(T^{g_k(j)}|T^{1:g_k(j)-1})\geq \log_2({q_{W_k}})-\delta_N \big{\}},\\
  \mathcal{L}^{(N)}_{T_{U'_k}|Y_2}&\triangleq \big{\{}j\in [N]:H(T^{g_k(j)}|Y_2^{1:N}, T^{1:g_k(j)-1})\leq \delta_N \big{\}}.
  \end{aligned}
  \end{equation}
  Since two senders' common messages are independent from each other, we have
  \begin{equation*}
  \mathcal{H}^{(N)}_{S_{U'_k}}=\mathcal{H}^{(N)}_{T_{U'_k}}=\mathcal{H}^{(N)}_{W_k},
  \end{equation*}
  where $\mathcal{H}^{(N)}_{W_k}\triangleq \big{\{}j\in [N]:H(U_k^{'j}|U_k^{'1:j-1})\geq \log_2({q_{W_k}})-\delta_N \big{\}}$.
  
  Define the following sets of indices for Sender 1,
  	\begin{equation}
  	\mathcal{C}^1_1\triangleq \mathcal{H}^{(N)}_{S_{U'_1}}\cap \mathcal{L}^{(N)}_{S_{U'_1}|Y_1},~~	\mathcal{C}^2_1\triangleq \mathcal{H}^{(N)}_{T_{U'_1}}\cap \mathcal{L}^{(N)}_{T_{U'_1}|Y_2},
  	\end{equation}
  and similarly define $\mathcal{C}^1_2$ and $\mathcal{C}^2_2$ for Sender 2.  
  From (\ref{MACR}) we have
  \begin{equation}
  \begin{aligned}
  \label{TyI-CR}
  \lim\limits_{N\rightarrow \infty} \frac{1}{N}|\mathcal{C}^1_1| &=\mathbf{P}^c(1),~~
  \lim\limits_{N\rightarrow \infty} \frac{1}{N}|\mathcal{C}^2_1| =\mathbf{\tilde{P}}^c(1)\geq \mathbf{P}^c(1),\\
  \lim\limits_{N\rightarrow \infty} \frac{1}{N}|\mathcal{C}^1_2| &=\mathbf{P}^c(2),~~
  \lim\limits_{N\rightarrow \infty} \frac{1}{N}|\mathcal{C}^2_2| =\mathbf{\tilde{P}}^c(2)\geq \mathbf{P}^c(2).
  \end{aligned}  
  \end{equation}
  Choose an arbitrary subset of $\mathcal{C}_1^2\setminus \mathcal{C}_1^1$, denoted as $\mathcal{C}_1^{21}$, such that $|\mathcal{C}_1^{21}|=|\mathcal{C}_1^1\setminus \mathcal{C}_1^2|$, and an arbitrary subset of $\mathcal{C}_2^2\setminus \mathcal{C}_2^1$, denoted as $\mathcal{C}_2^{21}$, such that $|\mathcal{C}_2^{21}|=|\mathcal{C}_2^1\setminus \mathcal{C}_2^2|$.
  Partition the indices of $U_1^{'1:N}$ as follows:
  \begin{figure}[tb]
  	\centering
  	\includegraphics[width=9cm]{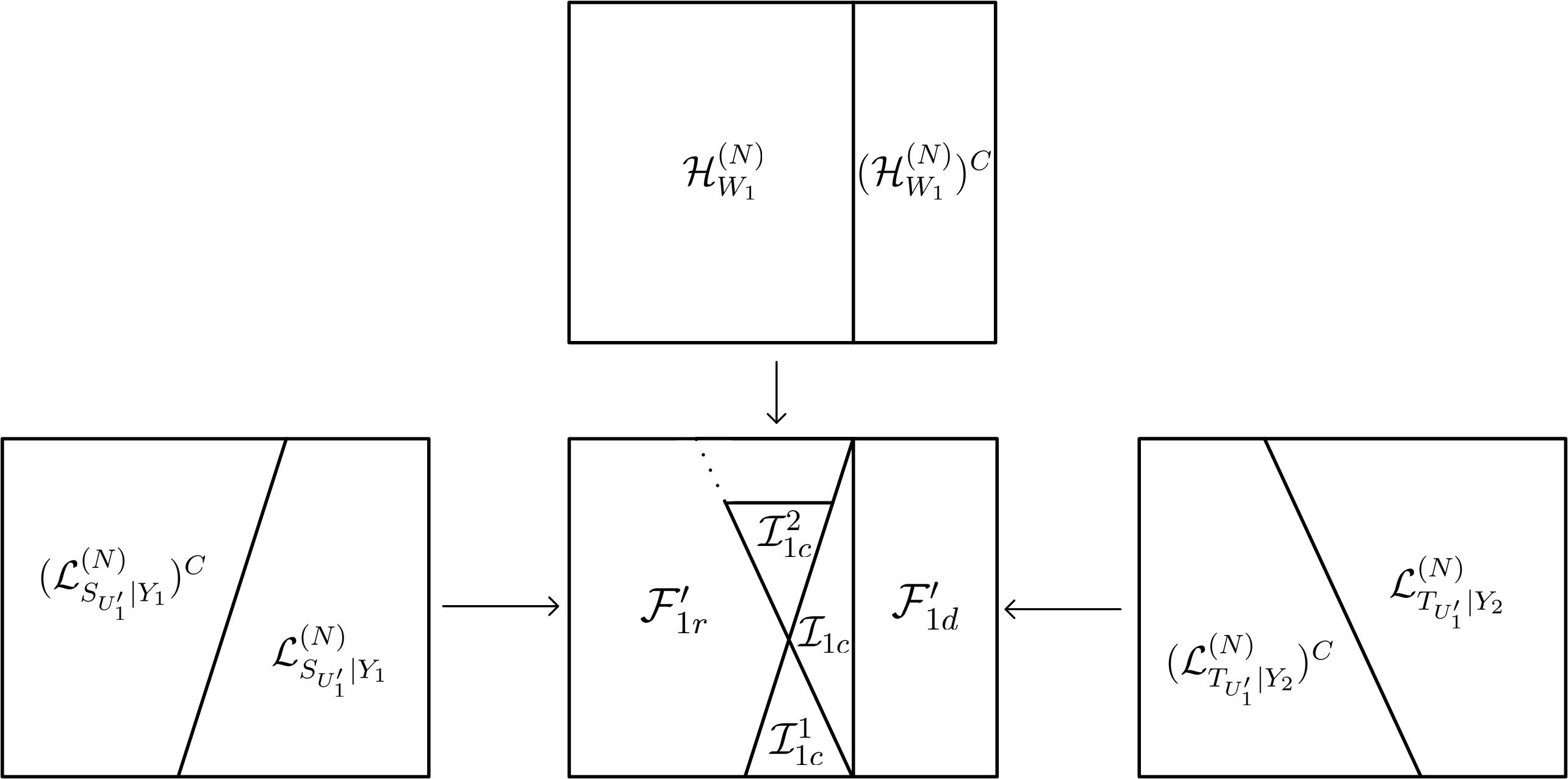}
  	\caption{Graphical representation of the partition for $U_1^{'1:N}$ of Type A points.} \label{fig:codecons-1c}
  \end{figure}
  \begin{equation}
  \begin{aligned}
  \mathcal{I}_{1c}&=\mathcal{C}_1^1\cap\mathcal{C}_1^2, ~~
  \mathcal{I}_{1c}^1=\mathcal{C}_1^1\setminus \mathcal{C}_1^2,~~
  \mathcal{I}_{1c}^2=\mathcal{C}_1^{21},\\
  \mathcal{F}'_{1r}&=\mathcal{H}^{(N)}_{W_1}\setminus (\mathcal{I}_{1c}\cup \mathcal{I}_{1c}^1\cup \mathcal{I}_{1c}^2), ~~
  \mathcal{F}'_{1d}=(\mathcal{H}^{(N)}_{W_1})^C,
  \end{aligned}
  \end{equation}
  as shown in Fig. \ref{fig:codecons-1c}, and similarly define $\mathcal{I}_{2c}$, $\mathcal{I}_{2c}^1$, $\mathcal{I}_{2c}^2$, $\mathcal{F}'_{2r}$ and $\mathcal{F}'_{2d}$ for Sender 2.

  \subsubsection{Partition Scheme for Type B Points}
  
  Let $\mathbf{P}$ be a point of Type B, $\mathbf{P}^c$ be the corresponding common message rate pair, and $\mathbf{P}^1$ and $\mathbf{P}^2$ be Receiver 1's and Receiver 2's rate pairs respectively. Without loss of generality, we consider the case when $\mathbf{P}^c\in \mathcal{R}(P_{Y_2|W_1W_2})\setminus \mathcal{R}(P_{Y_1|W_1W_2})$ and $\mathbf{P}^c(1)\leq I(W_1;Y_1)$ in this subsection as an example. In this case, Receiver 1 applies the Type II partially decoding while Receiver 2 adopts Type I.
  
  Choose $\mathbf{\bar{P}}^1=\big{(}I(X_1W_2;Y_1)-\mathbf{P}^1(2),\mathbf{P}^1(2)\big{)}$, which is on the dominant face of $\mathcal{R}(P_{Y_1|X_1W_2})$ and larger than $\mathbf{P}^1$, and $\mathbf{\tilde{P}}^c=\big{(}I(W_1W_2;Y_2)-\mathbf{P}^c(2),\mathbf{P}^c(2)\big{)}$, which is on the dominant face of $\mathcal{R}(P_{Y_2|W_1W_2})$ and larger than $\mathbf{P}^c$, as the target points for conducting monotone chain rule expansions in our code design. 
  Let $S^{1:2N}$ be the monotone chain rule expansion that achieves $\mathbf{\bar{P}}^1$ in $\mathcal{R}(P_{Y_1|X_1W_2})$, and $T^{1:2N}$ the expansion that achieves $\mathbf{\tilde{P}}^c$ in $\mathcal{R}(P_{Y_2|W_1W_2})$. Denote the sets of indices in $S^{1:2N}$ with $S^j\in U_1^{1:N}$ and $S^j\in U_2^{'1:N}$ by $\mathcal{S}_{U_1}$ and $\mathcal{S}_{U'_2}$ respectively, and those in $T^{1:2N}$ with $T^j\in U_1^{'1:N}$ and $T^j\in U_2^{'1:N}$ by $\mathcal{T}_{U'_1}$ and $\mathcal{T}_{U'_2}$ respectively. Let $f_1(j):[N]\rightarrow \mathcal{S}_{U_1}$ be the mapping from indices of $U_1^{1:N}$ to those of $S^{\mathcal{S}_{U_1}}$, $f_2(j):[N]\rightarrow \mathcal{S}_{U'_2}$ the mapping from indices of $U_2^{'1:N}$ to those of $S^{\mathcal{S}_{U'_2}}$, and $g_k(j):[N]\rightarrow \mathcal{T}_{U'_j}$ the mapping from indices of $U_k^{'1:N}$ to those of $T^{\mathcal{T}_{U'_k}}$ for $k=1,2$. Define $\mathcal{H}^{(N)}_{W_1}$, $\mathcal{H}^{(N)}_{W_2}$, $\mathcal{H}^{(N)}_{S_{U'_2}}$, $\mathcal{H}^{(N)}_{T_{U'_1}}$, $\mathcal{H}^{(N)}_{T_{U'_2}}$, $\mathcal{L}^{(N)}_{S_{U'_2}|Y_1}$, $\mathcal{L}^{(N)}_{T_{U'_1}|Y_2}$ and $\mathcal{L}^{(N)}_{T_{U'_2}|Y_2}$ in the same way as in the Type A case, and additionally define
  \begin{equation}
  \mathcal{L}^{(N)}_{W_1|Y_1}\triangleq \big{\{}j\in [N]:H(U_1^{'j}|Y_1^{1:N},U_1^{'1:j-1})\leq \delta_N \big{\}}.
  \end{equation}
  
  Define the following sets of indices for two senders:
  \begin{equation}
  \begin{aligned}
  &\mathcal{C}'_1\triangleq \mathcal{H}^{(N)}_{W_1}\cap \mathcal{L}^{(N)}_{W_1|Y_1},~~
  \mathcal{C}''_1\triangleq \mathcal{H}^{(N)}_{W_1}\cap \mathcal{L}^{(N)}_{T_{U'_1}|Y_2},\\
  &\mathcal{C}^1_2\triangleq \mathcal{H}^{(N)}_{W_2}\cap \mathcal{L}^{(N)}_{S_{U'_2}|Y_1},~~
  \mathcal{C}^2_2\triangleq \mathcal{H}^{(N)}_{W_2}\cap \mathcal{L}^{(N)}_{T_{U'_2}|Y_2},
  \end{aligned}
  \end{equation}
  which satisfy
  \begin{align*}
  &\lim\limits_{N\rightarrow \infty} \frac{1}{N}|\mathcal{C}'_1| =I(W_1;Y_1)\geq \mathbf{P}^c(1),~~
  \lim\limits_{N\rightarrow \infty} \frac{1}{N}|\mathcal{C}''_1| =\mathbf{\tilde{P}}^c(1)\geq \mathbf{P}^c(1),\\
  &\lim\limits_{N\rightarrow \infty} \frac{1}{N}|\mathcal{C}^1_2| =\mathbf{\bar{P}}^1(2)= \mathbf{P}^c(2),~~
  \lim\limits_{N\rightarrow \infty} \frac{1}{N}|\mathcal{C}^2_2| =\mathbf{\tilde{P}}^c(2)= \mathbf{P}^c(2).
  \end{align*}  
  If $\mathbf{P}^c(1)=I(W_1;Y_1)$, let $\mathcal{C}_1^1=\mathcal{C}'_1$. Otherwise choose a subset $\mathcal{C}_1^1$ of $ \mathcal{C}'_1$ such that $|\mathcal{C}_1^1|=N\mathbf{P}^c(1)$. Similarly, if $\mathbf{\tilde{P}}^c(1)=\mathbf{P}^c(1)$, let $\mathcal{C}_1^2=\mathcal{C}''_1$. Otherwise choose a subset $\mathcal{C}_1^2\subset \mathcal{C}''_1$ such that $|\mathcal{C}_1^2|=N\mathbf{P}^c(1)$. Partition the indices of $U_1^{'1:N}$ as follows:
  \begin{equation}
  \label{TypeII-CP}
  \begin{aligned}
  \mathcal{I}_{1c}&=\mathcal{C}_1^1\cap\mathcal{C}_1^2,~~
  \mathcal{I}_{1c}^1= \mathcal{C}_1^1\setminus \mathcal{C}_1^2, ~~
  \mathcal{I}_{1c}^2= \mathcal{C}_1^2\setminus \mathcal{C}_1^1, \\
  \mathcal{F}'_{1r}&=\mathcal{H}^{(N)}_{W_1}\setminus (\mathcal{I}_{1c}\cup \mathcal{I}_{1c}^1\cup \mathcal{I}_{1c}^2), ~~
  \mathcal{F}'_{1d}=(\mathcal{H}^{(N)}_{W_1})^C,
  \end{aligned}
  \end{equation}
  and similarly define $\mathcal{I}_{2c}$, $\mathcal{I}_{2c}^1$, $\mathcal{I}_{2c}^2$, $\mathcal{F}'_{2r}$ and $\mathcal{F}'_{2d}$ for Sender 2.

  \subsubsection{Chaining Scheme for Common Messages}  
 Suppose the number of chained blocks is $K$. 
  Let $\mathbf{F}_{1c}$,  $\mathbf{F}'_{1c}$ and $\mathbf{F}''_{1c}$ (resp. $\mathbf{F}_{2c}$, $\mathbf{F}'_{2c}$ and $\mathbf{F}''_{2c}$) be three random sequences of length $|\mathcal{F}'_{1r}|$, $|\mathcal{I}^1_{1c}|$ and $|\mathcal{I}^2_{1c}|$ (resp. $|\mathcal{F}'_{2r}|$, $|\mathcal{I}^1_{2c}|$ and $|\mathcal{I}^2_{2c}|$) respectively and uniformly distributed over $\mathcal{W}_1$ (resp. $\mathcal{W}_2$). Sender 1 encodes its common message as follows.
  
  (1) In Block 1, 
  \begin{itemize}
  	\item $\{u_1^{'j}\}_{j\in \mathcal{I}_{1c}\cup \mathcal{I}_{1c}^1}$ store common message symbols.
  	\item $\{u_1^{'j}\}_{j\in \mathcal{F}'_{1r}}=\mathbf{F}_{1c}$.
  	\item $\{u_1^{'j}\}_{j\in \mathcal{I}_{1c}^2}=\mathbf{F}''_{1c}$.
  	\item $\{u_1^{'j}\}_{j\in \mathcal{F}'_{1d}}$ are randomly generated according to conditional probability $P_{U_1^{'j}|U_1^{'1:j-1}}(u^{'j}_1|u_1^{'1:j-1})$.
  \end{itemize}  
  
  (2) In Block $i$ ($1<i<K$), 
  \begin{itemize}
  	\item $\{u_1^{'j}\}_{j\in \mathcal{I}_{1c}\cup \mathcal{I}_{1c}^1}$, $\{u_1^{'j}\}_{j\in \mathcal{F}'_{1r}}$ and $\{u_1^{'j}\}_{j\in \mathcal{F}'_{1d}}$ are determined in the same way as in Block 1.
  	\item $\{u_1^{'j}\}_{j\in \mathcal{I}_{1c}^2}$ are assigned to the same value as $\{u_1^{'j}\}_{j\in \mathcal{I}_{1c}^1}$ in Block $i-1$.
  \end{itemize}
  
  (3) In Block $K$, 
  \begin{itemize}
  	\item $\{u_1^{'j}\}_{j\in \mathcal{I}_{1c}}$, $\{u_1^{'j}\}_{j\in \mathcal{F}'_{1r}}$, $\{u_1^{'j}\}_{j\in \mathcal{I}_{1c}^2}$ and $\{u_1^{'j}\}_{j\in \mathcal{F}'_{1d}}$ are determined in the same way as in Block $i$ ($1<i<K$).
  	\item $\{u_1^{'j}\}_{j\in \mathcal{I}_{1c}^1}=\mathbf{F}'_{1c}$.
  \end{itemize}

  In each block, a vanishing fraction of the almost deterministic symbols, $\{u_1^{'j}\}_{j\in \mathcal{D}_1^1}$ and $\{u_1^{'j}\}_{j\in \mathcal{D}_1^2}$, are separately transmitted to Receiver 1 and 2 respectively with some reliable error-correcting code, where $\mathcal{D}_1^1=(\mathcal{H}_{W_1}^{(N)})^C\cap (\mathcal{L}_{S_{U'_1}|Y_1}^{(N)})^C$ in the Type A case and $\mathcal{D}_1^1=(\mathcal{H}_{W_1}^{(N)})^C\cap (\mathcal{L}_{W_1|Y_1}^{(N)})^C$ in the Type B case, and $\mathcal{D}_1^2=(\mathcal{H}_{W_1}^{(N)})^C\cap (\mathcal{L}_{T_{U'_1}|Y_2}^{(N)})^C$ in both cases. Note that random sequence $\mathbf{F}_{1c}$ is reused over $K$ blocks. Thus, the rate of frozen symbols that need to be shared between Sender 1 and Receiver 1 in the common message encoding, $\frac{1}{KN}(|\mathbf{F}_{1c}|+|\mathbf{F}'_{1c}|+|\mathbf{F}''_{1c}|)$, can be made negligible by increasing $K$. 
  
  Sender 2 encodes its common messages similarly by swapping subscripts 1 and 2.  
  
  \subsection{Private Message Encoding}
  \subsubsection{Partition Scheme for Type A Points}
  
  Define
  \begin{equation}
  \begin{aligned}
  \mathcal{H}^{(N)}_{X_1|W_1W_2}&\triangleq \big{\{}j\in[N]:H(U_1^{j}|U_1^{'1:N},U_2^{'1:N},U_1^{1:j-1})\geq \log_2({q_{X_1}})-\delta_N \big{\}},\\
  \mathcal{L}^{(N)}_{X_1|Y_1W_1W_2}&\triangleq \big{\{}j\in[N]:H(U_1^{j}|Y_1^{1:N},U_1^{'1:N},U_2^{'1:N},U_1^{1:j-1})\leq \delta_N \big{\}},
  \end{aligned}
  \end{equation}
  and similarly define $\mathcal{H}^{(N)}_{X_2|W_1W_2}$ and $\mathcal{L}^{(N)}_{X_2|Y_2W_1W_2}$. Due to the independence between two senders' messages, we have
  \begin{equation}
  \mathcal{H}^{(N)}_{X_1|W_1W_2}=\mathcal{H}^{(N)}_{X_1|W_1},~~
  \mathcal{H}^{(N)}_{X_2|W_1W_2}=\mathcal{H}^{(N)}_{X_2|W_2},
  \end{equation}
  where $\mathcal{H}^{(N)}_{X_k|W_k}\triangleq \big{\{}j\in[N]:H(U_i^{j}|U_k^{'1:N},U_k^{1:j-1})\geq \log_2({q_{X_k}})-\delta_N \big{\}}$ for $k=1,2$.
  Then define the following sets for $U_1^{1:N}$
  \begin{equation}
  \begin{aligned}
  \mathcal{I}_{1p}&\triangleq \mathcal{H}^{(N)}_{X_1|W_1W_2}\cap \mathcal{L}^{(N)}_{X_1|Y_1W_1W_2},\\
  \mathcal{F}_{1r}&=\mathcal{H}^{(N)}_{X_1|W_1W_2}\cap (\mathcal{L}^{(N)}_{X_1|Y_1W_1W_2})^C,\\
  \mathcal{F}_{1d}&=(\mathcal{H}^{(N)}_{X_1|W_1W_2})^C,\\
  \mathcal{D}_1&=(\mathcal{H}_{X_1|W_1W_2}^{(N)})^C\cap (\mathcal{L}_{X_1|Y_1W_1W_2}^{(N)})^C.
  \end{aligned}
  \end{equation}
  For $U_2^{1:N}$, $\mathcal{I}_{2p}$, $\mathcal{F}_{2r}$, $\mathcal{F}_{2d}$ and $\mathcal{D}_2$ are defined similarly.
  
  \subsubsection{Partition Scheme for Type B Points}
  From Definition \ref{def:TYII} we know that
  \begin{align}
  R_1^p&=\mathbf{\bar{P}}^1(1)-I(W_1;Y_1), \label{typeii-rp}\\
  R_2^p&=I(X_2;Y_2|W_1W_2). \label{typeii-rc2}
  \end{align}
  
  Define $\mathcal{H}^{(N)}_{X_1|W_1}$, $\mathcal{H}^{(N)}_{X_2|W_1W_2}$, $\mathcal{H}^{(N)}_{X_2|W_2}$ and $\mathcal{L}^{(N)}_{X_2|Y_2W_1W_2}$ in the same way as in the Type A case, and additionally define
  \begin{equation}
   \mathcal{L}^{(N)}_{S_{U_1}|Y_1W_1}\triangleq \big{\{}j\in [N]:H(S^{f_1(j)}|Y_1^{1:N},U_1^{'1:N}, S^{1:f_1(j)-1})\leq \delta_N \big{\}}.
  \end{equation}
  Then define $\mathcal{I}_{2p}$, $\mathcal{F}_{2r}$, $\mathcal{F}_{2d}$ and $\mathcal{D}_2$ for $U_2^{1:N}$ in the same way as in the Type A case, and define
  \begin{equation}
  \begin{aligned}
  \mathcal{I}_{1p}&= \mathcal{H}^{(N)}_{X_1|W_1}\cap \mathcal{L}^{(N)}_{S_{U_1}|Y_1W_1},\\
  \mathcal{F}_{1r}&=\mathcal{H}^{(N)}_{X_1|W_1}\cap (\mathcal{L}^{(N)}_{S_{U_1}|Y_1W_1})^C,\\
  \mathcal{F}_{1d}&=(\mathcal{H}^{(N)}_{X_1|W_1})^C,\\
  \mathcal{D}_1&=(\mathcal{H}_{X_1|W_1}^{(N)})^C\cap (\mathcal{L}_{S_{U_1}|Y_1W_1}^{(N)})^C,
  \end{aligned}
  \end{equation}
  for $U_1^{1:N}$. Note that the permutation $S^{1:2N}$ is chosen to achieve $\mathbf{\bar{P}}^1$ in Receiver 1's effective channel $P_{Y_1|X_1W_2}$ without the knowledge of $W_1$, but the code construction for $U_1^{1:N}$ is determined jointly by this permutation and the side information of $W_1^{1:N}$.

  \subsubsection{Encoding for Private Messages}

  Let $\mathbf{F}_{1p}$ (resp. $\mathbf{F}_{2p}$) be a random sequence of length $|\mathcal{F}_{1p}|$ (resp. $|\mathcal{F}_{2p}|$) and uniformly distributed over $\mathcal{X}_1$ (resp. $\mathcal{X}_2$).
  Sender 1 encodes its private message in each block as follows.
   
  \begin{itemize}
  	\item $\{u_1^{j}\}_{j\in \mathcal{I}_{1p}}$ store private message symbols.
  	\item $\{u_1^{j}\}_{j\in \mathcal{F}_{1r}}=\mathbf{F}_{1p}$.
  	\item $\{u_1^{j}\}_{j\in \mathcal{F}_{1d}}$ are randomly generated according to probability $P_{U_1^j|U_1^{1:N}U_1^{1:j-1}}(u_1^j|u_1^{'1:N},u_1^{1:j-1})$.
  	\item $\{u_1^{j}\}_{j\in \mathcal{D}_1}$ are separately transmitted to Receiver 1 with some reliable error-correcting code.
  \end{itemize}

  Sender 2 encodes its private message similarly by swapping subscripts 1 and 2.  
  Note that random sequence $\mathbf{F}_{1p}$ and $\mathbf{F}_{2p}$ are reused over $K$ blocks. Thus, the rate of frozen symbols in the private message encoding can also be made negligible by increasing $K$.
  
  \subsection{Decoding}
  \subsubsection{Decoding for Type A Points}
  
  Receiver 1 decodes two senders' common messages from Block 1 to Block $K$. 
  \begin{itemize}
  	\item In Block 1, for $k=1,2$,
  	\begin{equation}
  	\bar{u}_k^{'j}=
  	\begin{cases}
  	u_k^{'j},  &\text{if } j\in (\mathcal{L}^{(N)}_{S_{U'_k}|Y_1})^C \\
  	\arg\max_{u\in\{0,1\}}P_{S^{f_k(j)}|Y_k^{1:N}S^{1:f_k(j)-1}}(u|y_k^{1:N},s^{1:f_k(j)-1}),&\text{if } j\in \mathcal{L}^{(N)}_{S_{U'_k}|Y_1}
  	\end{cases}
  	\end{equation}

  	\item In Block $i$ $(1<i<K)$, $\{\bar{u}_1^{'j}\}_{j\in \mathcal{I}_{1c}^2}$ and $\{\bar{u}_2^{'j}\}_{j\in \mathcal{I}_{2c}^2}$ are deduced from $\{\bar{u}_1^{'j}\}_{j\in \mathcal{I}_{1c}^1}$ and $\{\bar{u}_2^{'j}\}_{j\in \mathcal{I}_{2c}^1}$ in Block $i-1$ respectively, and the rest are decoded in the same way as in Block 1.
  	
  	\item In Block $K$, $\{\bar{u}_1^{'j}\}_{j\in \mathcal{I}_{1c}^1}$ and $\{\bar{u}_2^{'j}\}_{j\in \mathcal{I}_{2c}^1}$ are assigned to the pre-shared value between Sender 1 and the two receivers, and the rest are decoded in the same way as in Block $i$ $(1<i<K)$. 
  \end{itemize}
  
  Having recovered the common messages in a block, Receiver 1 decodes its private message in that block as
  \begin{equation}
  \bar{u}_1^{j}=
  \begin{cases}
  u_1^j,&\text{if } j\in (\mathcal{L}_{X_1|Y_1W_1W_2}^{(N)})^C\\
  \arg\max_{u\in\{0,1\}}P_{U_1^{j}|Y_1^{1:N}U_1^{'1:N}U_2^{'1:N}U_1^{1:j-1}}(u|y_1^{1:N},\bar{u}_1^{'1:N},\bar{u}_2^{'1:N},u_1^{1:j-1}),&\text{if } i\in \mathcal{L}_{X_1|Y_1W_1W_2}^{(N)}
  \end{cases}
  \end{equation}
  
  Receiver 2 decodes similarly, except that it decodes from Block $K$ to Block 1.

  \subsubsection{Decoding for Type B Points}
  
  Receiver 1 decodes from Block 1 to Block $K$.  
  \begin{itemize}
  	\item In Block 1, Sender 1 first decodes its intended common message as
  	\begin{equation}
  	\bar{u}_1^{'j}=
  	\begin{cases}
  	u_1^{'j},&\text{if } j\in (\mathcal{L}_{W_1|Y_1}^{(N)})^C\\
  	\arg\max_{u\in\{0,1\}}P_{U_1^{'j}|Y_1^{1:N}U_1^{'1:j-1}}(u|y_1^{1:N},\bar{u}_1^{'1:j-1}),&\text{if } i\in \mathcal{L}_{W_1|Y_1}^{(N)}
  	\end{cases}
  	\end{equation}
  	Then it decodes its private message and Sender 2's common message jointly as
  	\begin{equation}
  	\bar{u}_1^{j}=
  	\begin{cases}
  	u_1^{j},  &\text{if } j\in (\mathcal{L}^{(N)}_{S_{U_1}|Y_1})^C \\
  	\arg\max_{u\in\{0,1\}}P_{S^{f_1(j)}|Y_1^{1:N}U_1^{'1:N}S^{1:f_1(j)-1}}(u|y_1^{1:N},\bar{u}_1^{'1:N},s^{1:f_1(j)-1}),&\text{if } j\in \mathcal{L}^{(N)}_{S_{U_1}|Y_1}
  	\end{cases}
  	\end{equation}
  	\begin{equation}
  	\bar{u}_2^{'j}=
  	\begin{cases}
  	u_2^{'j},  &\text{if } j\in (\mathcal{L}^{(N)}_{S_{U'_2}|Y_1})^C \\
  	\arg\max_{u\in\{0,1\}}P_{S^{f_2(j)}|Y_1^{1:N}S^{1:f_2(j)-1}}(u|y_1^{1:N},s^{1:f_2(j)-1}),&\text{if } j\in \mathcal{L}^{(N)}_{S_{U'_2}|Y_1}
  	\end{cases}
  	\end{equation}
  	
  	\item In Block $i$ $(1<i<K)$, $\{\bar{u}_1^{'j}\}_{j\in \mathcal{I}_{1c}^2}$ and $\{\bar{u}_2^{'j}\}_{j\in \mathcal{I}_{2c}^2}$ are deduced from $\{\bar{u}_1^{'j}\}_{j\in \mathcal{I}_{1c}^1}$ and $\{\bar{u}_2^{'j}\}_{j\in \mathcal{I}_{2c}^1}$ in Block $i-1$ respectively, and the rest are decoded in the same way as in Block 1.
  	
  	\item In Block $K$, $\{\bar{u}_1^{'j}\}_{j\in \mathcal{I}_{1c}^1}$ and $\{\bar{u}_2^{'j}\}_{j\in \mathcal{I}_{2c}^1}$ are assigned to the pre-shared value between Sender 1 and the two receivers, and the rest are decoded in the same way as in Block $i$ $(1<i<K)$. 
  \end{itemize}
  
  Receiver 2 decodes from Block $K$ to Block 1 in the same way as in the Type A scheme.

  \subsection{Code Construction}
  \label{S:CodeCons}
  As pointed out by a reviewer of this paper, existing efficient construction algorithms (such as \cite{tal2013construct}) for point-to-point polar codes may not be directly applied to the permutation based MAC polar codes in the general case, as the permutation introduces a random variable that involves a complicated relation with the original pair of random variables. Thus, it is currently not clear how much the code construction complexity of the permutation based MAC polar codes is. Nevertheless, as has been shown in \cite{arikan2012sw}, permutations of type $0^i1^N0^{N-i}$ ($0\leq i\leq N$) are sufficient to achieve the whole achievable rate region of a 2-user MAC. In this subsection we show how to construct this kind of MAC polar codes with the approximation method of \cite{tal2013construct}.
  
  Let $W(Y|X_1,X_2)$ be a discrete memoryless MAC with $X_1\in\mathcal{X}_1$ and $X_2\in\mathcal{X}_2$. Define $U_1^{1:N}=X_1^{1:N}\mathbf{G}_N$, $U_2^{1:N}=X_2^{1:N}\mathbf{G}_N$, and $S^{1:2N}=U_1^{1:i_m}U_2^{1:N}U_1^{i_m+1:N}$ ($0\leq i_m\leq N$)\footnote{In Algorithm \ref{Algo-1} we have restricted $i_m$ to be chosen from $[1,N]$ because $i_m-1$ must be non-negative. The case when $i_m=0$ is the same as that of $i_m=1$ except that $u_1^1$ needs to be averaged out in the end.}. In this case, the polarization of user 1's first $i_m$ synthesized channels is the same as that in the equivalent point-to-point channel when user 2's signal is treated as noise, and the polarization of user 1's last $N-i_m$ synthesized channels is the same as that in the equivalent point-to-point channel when user 2's signal is treated as side information. Thus, the method of \cite{tal2013construct} can be directly applied (one can also use the proposed Algorithm 1 by swapping the roles of the two users and considering the special cases of $i_m=0$ and $i_m=N$). For user 2, to apply the method of \cite{tal2013construct}, the recursive channel transformations need to be modified accordingly. Define the following two types of channel transformations for $W$:
  \begin{align*}
  &W\boxplus W(y^{1:2},u_1^2|u_1^1,u_2^1)= \sum_{u_2^2\in \mathcal{X}_2} W(y^1|u_1^1\oplus u_1^2,u_2^1\oplus u_2^2)W(y^2|u_1^2,u_2^2)P(u_1^2)P(u_2^2) \\
  &W\boxtimes W(y^{1:2},u_1^1,u_2^1|u_1^2,u_2^2)= W(y^1|u_1^1\oplus u_1^2,u_2^1\oplus u_2^2)W(y^2|u_1^2,u_2^2) P(u_1^1)P(u_2^1)
  \end{align*}
  Based on the channel degrading and upgrading method of \cite[Algoritm A and B]{tal2013construct}, we propose a constructing method for this type of MAC polar codes as shown in Algorithm \ref{Algo-1}. The $\text{\emph{degrading\_merge}}(W,\mu)$ (resp. $\text{\emph{upgrading\_merge}}(W,\mu)$) operation in this algorithm is to produce a degraded (resp. an upgraded) version of $W$, whose output alphabet size is at most $\mu$ so that it can be estimated at affordable cost. We only present a general picture of how this algorithm goes here. For details about the degrading and upgrading procedures, we refer the readers to \cite{tal2013construct,Tal2012const,Pereg2017upgrade}.
  
  \begin{algorithm}
  	\label{Algo-1}
  	\caption{Channel degrading/upgrading procedure for the MAC} 
  	{\bf Input:} 
  	An underlying MAC $W$, a bound $\mu =2\nu$ on the output alphabet size, a code length $N=2^n$, an index $i_d$ ($0\leq i_d \leq N-1$) with binary representation $i_d=\left\langle a_1a_2...a_n\right\rangle_2$ representing the index of a synthesized channel, and an index $i_m$ ($1\leq i_m \leq N$) with $i_m-1=\left\langle b_1b_2...b_n\right\rangle_2$ representing the permutation type.\\
  	{\bf Output:} 
  	A DMC that is degraded/upgraded with respect to the $(i_d+1)$th synthesized channel of user 2.
  	
  	$Q_W\leftarrow\text{\emph{degrading\_merge}}(W,\mu)/\text{\emph{upgrading\_merge}}(W,\mu)$,	\\
  	$s=0$,\\
  	\For{j=1,2,...,n}{
  		\eIf{$a_j=0$}{
  			$W_a=Q_W\boxplus Q_W$,
  		}{
  			$W_a=Q_W\boxtimes Q_W$,
  		}
  		\If{$b_j=0$}{
  			$i_t=\left\langle b_1b_2...b_{j-1}1\right\rangle_2$,\\
  			\eIf{$i_t\neq \left\langle a_1...a_j\right\rangle_2$}{
  			$W_a\overset{(a)}{=}\sum_{u_1^{i_t+1}\in \mathcal{X}_1}W_a$,
  	    	}{
  	    	  $s=1$;
  	        }
  		}
  	\If{$s=1$}{
  		$i_s=\left\langle a_1...a_{j-1}\bar{a}_j\right\rangle_2$, where $\bar{a}_j=a_j\oplus1$\\
  		$W_a\overset{(b)}{=}\sum_{u_1^{i_s+1}\in \mathcal{X}_1}W_a$,\\
  		$s=0$
  	    }
  		$Q_W\leftarrow\text{\emph{degrading\_merge}}(W_a,\mu)/\text{\emph{upgrading\_merge}}(W_a,\mu)$,\\
  	}  
  	\eIf{$i_d+1 > i_m$}{
  		$Q_{i_d+1}(\mathbf{\tilde{y}}|u_2^{i_d+1})=\sum_{u_1^{i_d+1}\in \mathcal{X}_1}Q_W(\mathbf{\tilde{y}}|u_1^{i_d+1},u_2^{i_d+1})P(u_1^{i_d+1})$
  	}{
  		$Q_{i_d+1}(\mathbf{\tilde{y}}|u_2^{i_d+1})=Q_W(\mathbf{\tilde{y}}|u_1^{i_d+1},u_2^{i_d+1})P(u_1^{i_d+1})$
  	}
  	\Return $Q_{i_d+1}$
  \end{algorithm}
  
  The idea of Algorithm \ref{Algo-1} is to first approximate the synthetic MACs in the recursive process and then synthesize the desired DMC in the end. In this algorithm, the purpose of (a) is to reduce the subsequent computations since $u_1^{i_p}$ with $i_p=\left\langle b_1...b_{j-1}1b'_{j+1}...b'_n\right\rangle_2$ will not be shown in the channel outputs of user 2's synthesized channels if $b_j=0$, where $b'_{j+1}...b'_n$ is any binary sequence of length $n-j$. (b) is to handle a special case when $b_j=0$ but $u_1^{i_t+1}$ happens to be a channel input. In this case we will have to eliminate the redundant channel output $u_1^{i_s+1}$ in the next stage. Note that the procedure of (b) only needs to be executed once at most. $\mathbf{\tilde{y}}$ denotes the output of channel $Q_W$.  Although Algorithm \ref{Algo-1} has some extra computations compared to \cite[Algoritm A and B]{tal2013construct}, the time complexity to evaluate all $N$ synthesized channels can still be reduced to $O(N)$ by sharing intermediate calculations between different synthesized channels. 
  
  Table \ref{example:construct} shows an example of the recursive process when $n=5$, $i_d=13$ (i.e., $\left\langle a_1a_2...a_n\right\rangle_2=\left\langle 01101\right\rangle_2$), $i_m=7$ (i.e., $\left\langle b_1b_2...b_n\right\rangle_2=\left\langle 00110\right\rangle_2$), in which we have ignored the channel degrading/upgrading procedure and only demonstrated the evolvement of the synthesized channels. In this example, $s=1$ is triggered in the $j=2$ stage, so in the next stage $u_1^3$ ($i_s=\left\langle 010\right\rangle_2$) is averaged out.
  
  \begin{table}[htbp]
  	\centering
  	\caption{$n=5$, $i_d=13$ ($\left\langle a_1a_2...a_n\right\rangle_2=\left\langle 01101\right\rangle_2$), $i_m=7$ ($\left\langle b_1b_2...b_n\right\rangle_2=\left\langle 00110\right\rangle_2$)}
  	\label{example:construct}
  	\begin{tabular}{|c|c|c|c|c|c|}
  		\hline
  		$j$& $a_j$ & $b_j$ & $i_t+1$ & $i_s+1$ & $W_a$\\
  		\hline
  		1& 0 & 0 & 2 & $\setminus$ & $W(y^{1:2}|u_1^1,u_2^1)$ \\
  		2& 1 & 0 & 2 & $\setminus$ & $W(y^{1:4},u_1^1,u_2^1|u_1^2,u_2^2)$ \\
  		3& 1 & 1 & $\setminus$ & 3 & $W(y^{1:8},u_1^{1:2},u_2^{1:3}|u_1^4,u_2^4)$\\
  		4& 0 & 1 & $\setminus$ & $\setminus$ & $W(y^{1:16},u_1^{1:4},u_2^{1:6}|u_1^7,u_2^7)$\\
  		5& 1 & 0 & 8 & $\setminus$ & $W(y^{1:32},u_1^{1:7},u_2^{1:13}|u_1^{14},u_2^{14})$ \\
  		\hline
  	\end{tabular}
  \end{table}

  \section{Performance Analysis}
  \label{S:PA}
  
  \subsection{Achievable Rates}
  \subsubsection{Type A Scheme}
  In the Type A scheme, the common message rates of the two senders in this scheme are
  \begin{equation}
  \begin{aligned}
  R_1^c&=\frac{K|\mathcal{I}_{1c}|+(K-1)|\mathcal{I}_{1c}^1|}{KN}=\frac{|\mathcal{C}_1^1|}{N}-\frac{|\mathcal{I}_{1c}^1|}{KN},\\
  R_2^c&=\frac{K|\mathcal{I}_{2c}|+(K-1)|\mathcal{I}_{2c}^1|}{KN}=\frac{|\mathcal{C}_2^1|}{N}-\frac{|\mathcal{I}_{2c}^1|}{KN}.
  \end{aligned}
  \end{equation}
  From (\ref{TyI-CR}) we have
  \begin{equation}
  \lim\limits_{N\rightarrow \infty,K\rightarrow \infty} R_1^c=\mathbf{P}^c(1),~~
  \lim\limits_{N\rightarrow \infty,K\rightarrow \infty} R_2^c=\mathbf{P}^c(2).  \label{TyI-Rc}
  \end{equation}
  The private message rates of the two senders are
  \begin{equation}
  R_1^p=\frac{1}{N}|\mathcal{I}_{1p}|,~~
  R_2^p=\frac{1}{N}|\mathcal{I}_{2p}|.
  \end{equation}
  Since the private message encoding is just standard point-to-point polar coding, we have
  \begin{equation}
  \lim\limits_{N\rightarrow \infty} R_1^p=I(X_1;Y_1|W_1W_2),~~
  \lim\limits_{N\rightarrow \infty} R_2^p=I(X_2;Y_2|W_1W_2). \label{TyI-Rp}
  \end{equation}
  Thus, our proposed scheme achieves the target Type A point $\mathbf{P}$.
  
  \subsubsection{Type B Scheme}
  In the Type B scheme, the common message rates can also be written as
  \begin{equation}
  R_1^c=\frac{|\mathcal{C}_1^1|}{N}-\frac{|\mathcal{I}_{1c}^1|}{KN},~~
  R_2^c=\frac{|\mathcal{C}_2^1|}{N}-\frac{|\mathcal{I}_{2c}^1|}{KN},
  \end{equation}
  with
  \begin{equation*}
  \lim\limits_{N\rightarrow \infty,K\rightarrow \infty} R_1^c=\mathbf{P}^c(1),~~
  \lim\limits_{N\rightarrow \infty,K\rightarrow \infty} R_2^c=\mathbf{P}^c(2).
  \end{equation*}
  Same as in the Type A case, the private message rate of Sender 2 achieves (\ref{typeii-rc2}). For Sender 1's private message rate, the following lemma shows that our proposed scheme achieves (\ref{typeii-rp}).  
  \begin{lemma}
  	\label{lemma.I1p}
  	$\lim\limits_{N\rightarrow\infty}\frac{1}{N}|\mathcal{I}_{1p}|=\mathbf{\bar{P}}^1(1)-I(W_1;Y_1)$.
  \end{lemma} 
  \begin{proof}
  	See Appendix \ref{APPEN-C}.  
  \end{proof}
  
  \subsection{Total Variation Distance}
  Let $P_U(u)$ denote the target distribution of random variable $U$, $Q_U(u)$ denote the induced distribution of $U$ by our encoding scheme, and $\parallel P-Q \parallel$ denote the total variation distance between distributions $P$ and $Q$. We have the following lemma.
  \begin{lemma}
  	\label{lemma.TVD}
  	For $i\in[1,K]$,
  	\begin{align}
  	&~~\parallel P_{W_1^{1:N}W_2^{1:N}X_1^{1:N}X_2^{1:N}Y_1^{1:N}Y_2^{1:N}}-Q_{(W_1^{1:N}W_2^{1:N}X_1^{1:N}X_2^{1:N}Y_1^{1:N}Y_2^{1:N})_i}\parallel \leq 4\sqrt{\log 2}\sqrt{N\delta_N}, \label{VD}
  	\end{align}
  	where $(\cdot)_i$ stands for random variables in Block $i$ ($1\leq i\leq K$).
  \end{lemma} 
  \begin{proof}
  	See Appendix \ref{APPEN-D}.  
  \end{proof}

  \subsection{Error Performance}
  \label{S:ErrorPerf}
  \begin{lemma}
  	\label{lemma.EP}
  	The error probability of a receiver with the Type I partially-joint decoding in the overall $K$ blocks can be upper bounded by
  	\begin{equation}
  	P_{e}^I\leq \frac{(K+1)(K+2)}{2}N\delta_N+2K(K+1)\sqrt{\log 2}\sqrt{N\delta_N},
  	\end{equation}
  	while error probability of a receiver with the Type II partially-joint decoding in the overall $K$ blocks can be upper bounded by
  	\begin{equation}
  	P_{e}^{II}\leq \frac{K(K+1)(K+5)}{6}N\delta_N+\frac{2K(K^2+6K-1)}{3}\sqrt{\log 2}\sqrt{N\delta_N}.
  	\end{equation}
  \end{lemma} 
  \begin{proof}
  	See Appendix \ref{APPEN-E}.  
  \end{proof}

  We can see that the chaining scheme has a more detrimental effect on the Type II decoding than on the Type I one. This is because in the Type I decoding, only the common message decoding stage involves chaining, while in the Type II decoding, both stages of decoding involve chaining.
  
  \subsection{Complexity}
  Since both our scheme and the scheme of \cite{wang2015channel} use the monotone chain rule based MAC polar codes, their encoding and decoding complexities are similar. As we have discussed in Section \ref{S:CodeCons}, our proposed polar codes can be constructed with complexity $O(N)$. Note that it is not clear whether the 3-user MAC polar codes used in \cite{wang2015channel} can also be constructed in a similar way. Therefore the construction complexity of our scheme is smaller than that in \cite{wang2015channel} (at least equal if the permutation based $m$-user MAC polar codes can be constructed at complexity $O(mN)$). 
   
  In the rest of this subsection, we discuss another simplification of our proposed scheme compared to \cite{wang2015channel}, i.e., the complexity reduction in the design of ARVs.  
  The Han-Kobayashi region is expressed with ARVs. Finding suitable ARVs to achieve a target rate pair is in fact part of the code design, since unlike the channel statistics which are given, the ARVs need to be designed and optimized. 
  Consider a 2-user DM-IC $P_{Y_1Y_2|X_1X_2}(y_1,y_2|x_1,x_2)$ and fixed ARV alphabets $\mathcal{W}_1$, $\mathcal{W}_2$, $\mathcal{V}_1$ and $\mathcal{V}_2$.
  Denote $\mathcal{P}_{W_1}$, $\mathcal{P}_{W_2}$, $\mathcal{P}_{V_1}$ and $\mathcal{P}_{V_2}$ as the sets of distributions $P_{W_1}$, $P_{W_2}$, $P_{V_1}$ and $P_{V_2}$, respectively, $\mathcal{P}_{X_1|W_1V_1}$ and $\mathcal{P}_{X_2|W_2V_2}$ the sets of deterministic mappings $P_{X_1|W_1V_1}$ and $P_{X_2|W_2V_2}$, respectively, and $\mathcal{P}_{X_1W_1}$ and $\mathcal{P}_{X_2W_2}$ the sets of joint distributions $P_{X_1W_1}$ and $P_{X_2W_2}$, respectively. 
  Since $P_{X_1|W_1V_1}$ and $P_{X_2|W_2V_2}$ equal either 0 or 1, it is easy to see that
  \begin{align*}
  |\mathcal{P}_{X_1|W_1V_1}|=2^{|\mathcal{W}_1|\cdot|\mathcal{V}_1|\cdot|\mathcal{X}_1|},~~ 
  |\mathcal{P}_{X_2|W_2V_2}|=2^{|\mathcal{W}_2|\cdot|\mathcal{V}_2|\cdot|\mathcal{X}_2|}.
  \end{align*}
  
  It is impractical to evaluate all the distributions in the aforementioned sets, thus certain quantization is needed. As an example, we assume that probabilities can only be chosen from a quantized subset of $[0,1]$, say $\mathcal{P}_Q$, and define $\mathcal{P}^*_{W_1}$, $\mathcal{P}^*_{W_2}$, $\mathcal{P}^*_{V_1}$, $\mathcal{P}^*_{V_2}$, $\mathcal{P}^*_{X_1W_1}$ and $\mathcal{P}^*_{X_2W_2}$ respectively as the subsets of $\mathcal{P}_{W_1}$, $\mathcal{P}_{W_2}$, $\mathcal{P}_{V_1}$, $\mathcal{P}_{V_2}$, $\mathcal{P}_{X_1W_1}$ and $\mathcal{P}_{X_2W_2}$ when the probabilities are restricted to be chosen from $\mathcal{P}_Q$.  
  To evaluate the original Han-Kobayashi region, the number of calculations for the region of $\mathcal{R}_{HK}^{o}(P^*)$ defined in Theorem \ref{theorem:HK-o} is $2^{|\mathcal{W}_1|\cdot|\mathcal{V}_1|\cdot|\mathcal{X}_1|+|\mathcal{W}_2|\cdot|\mathcal{V}_2|\cdot|\mathcal{X}_2|}\cdot|\mathcal{P}^*_{W_1}|\cdot|\mathcal{P}^*_{W_2}|\cdot|\mathcal{P}^*_{V_1}|\cdot|\mathcal{P}^*_{V_2}|$.
  Meanwhile, to evaluate the compact Han-Kobayashi region, the number of calculations for the region of $\mathcal{R}_{HK}(P_1^*)$ defined in Theorem \ref{theorem:HK} is $|\mathcal{P}^*_{X_1W_1}|\cdot|\mathcal{P}^*_{X_2W_2}|$. As long as $|\mathcal{X}_1|\leq |\mathcal{V}_1|$ and $|\mathcal{X}_2|\leq |\mathcal{V}_2|$, $|\mathcal{P}^*_{X_1W_1}|$ and $|\mathcal{P}^*_{X_2W_2}|$ will not be larger than $|\mathcal{P}^*_{W_1}|\cdot|\mathcal{P}^*_{V_1}|$ and $|\mathcal{P}^*_{W_2}|\cdot|\mathcal{P}^*_{V_2}|$, respectively. Due to this fact and that the expressions for $\mathcal{R}_{HK}(P_1^*)$ is much simpler than those of $\mathcal{R}_{HK}^{o}(P^*)$, we can conclude that our proposed scheme only requires $\frac{1}{2^{|\mathcal{W}_1|\cdot|\mathcal{V}_1|\cdot|\mathcal{X}_1|+|\mathcal{W}_2|\cdot|\mathcal{V}_2|\cdot|\mathcal{X}_2|}}$ computation (at most) compared to the scheme of \cite{wang2015channel} in the design of ARVs. This can be quite a complexity reduction, especially for large alphabet size cases.

  \section{Extension to Interference Networks}
  \label{S:IN}
  
  So far we have shown that our proposed two types of partially decoding schemes can achieve the Han-Kobayashi region of the 2-user DM-IC via both random coding and polar coding. A natural question is whether they can be extended to arbitrary DM-INs. In this section, we show that partially-joint decoding can also make heterogeneous superposition polar coding schemes easier to realize in DM-INs.

  A $K$-sender $L$-receiver DM-IN, denoted by $(K,L)$-DM-IN, consists of $K$ senders and $L$ receivers. Each sender $k\in[K]$ transmits an independent message $M_k$ at rate $R_k$, while each receiver $l\in[L]$ wishes to recover a subset $\mathcal{D}_l\subset[K]$ of the messages. Similar to the Han-Kobayashi strategy in the 2-user DM-IC, $M_k$ can be split into several component messages, each intended for a group of receivers. If a message is intended for only one receiver, we refer to it as a private message. Otherwise we refer to it as a common message. We only consider the case when each sender has only one private message intended for some receiver and (possibly) multiple common messages intended also for this receiver. More complicated cases can be resolved by decomposing a sender with multiple private and common messages into a certain number of virtual senders of this type.
  
  \begin{figure}[tb]
  	\centering
  	\includegraphics[width=7cm]{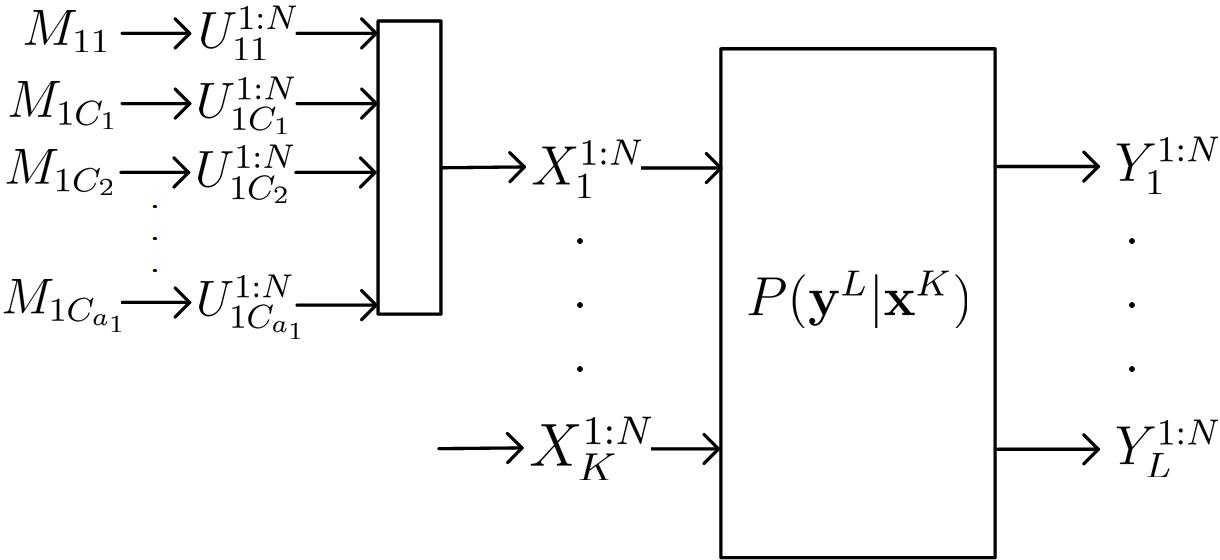}
  	\caption{Sender 1's part in the equivalent channel of the $(K,L)$-DM-IN.} \label{fig:in-hk-o}
  \end{figure}
  Fig. \ref{fig:in-hk-o} shows Sender 1's part of the equivalent channel of the $(K,L)$-DM-IN with a private message $M_{11}$ intended for Receiver 1, and common messages $M_{1C_{1}}$, $M_{1C_{2}}$, ..., $M_{1C_{a_1}}$ ($a_1\geq 1$) intended for Receiver 1 and some other receiver groups. It is shown in \cite{bandemer2015IN} that the optimal achievable rate region when the encoding is restricted to random coding ensembles is the intersection of rate regions for its component multiple access channels in which each receiver recovers its private message as well as its common messages. Thus, one can design a code for the compound MAC to achieve the optimal rate region, which belongs to the homogeneous superposition variant and has been realized by polar codes in \cite{wang2015channel}. Here we discuss using the proposed partially-joint decoding idea to design a heterogeneous one.

  \begin{figure}[tb]
  	\centering
  	\includegraphics[width=7cm]{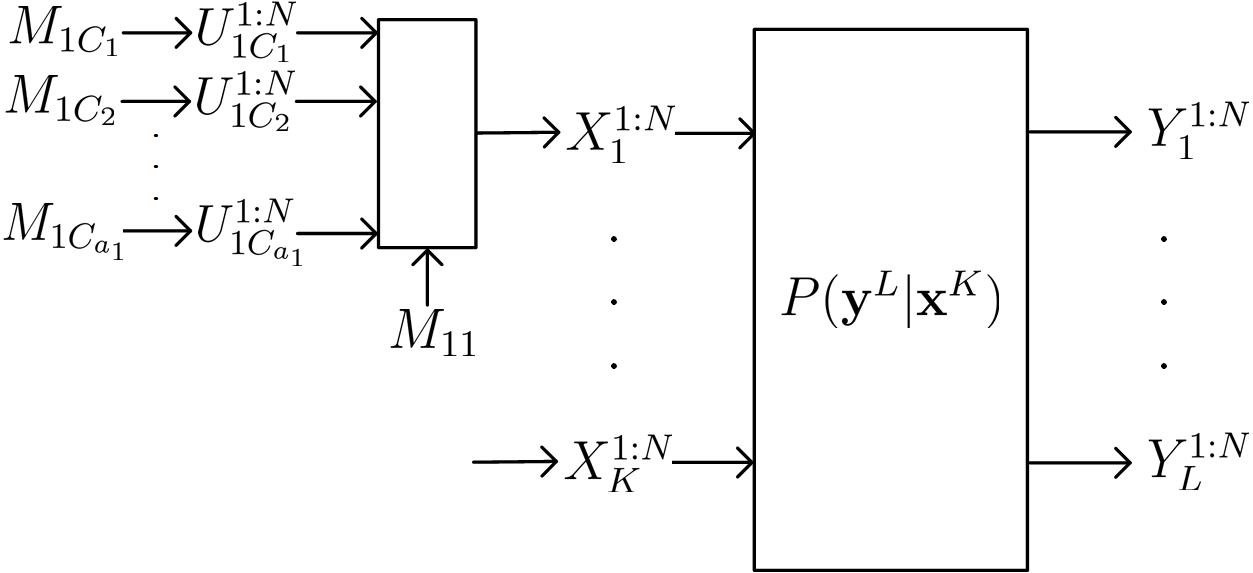}
  	\caption{Sender 1's part in the equivalent channel of the $(K,L)$-DM-IN with the proposed approach.} \label{fig:in-hk}
  \end{figure}
  
  Firstly, consider the case when only Sender 1 uses the heterogeneous approach. Instead of generating a codeword for each message and then merging them with some mapping function as in Fig. \ref{fig:in-hk-o}, now we generate codewords for common messages first and then encode them together with private message $M_{11}$ via superposition coding, as shown in Fig. \ref{fig:in-hk}. Let $P^*(X_1|U_{11},U_{1C_1},...,U_{1C_{a_1}})$ be the deterministic mapping from $U_{11},U_{1C_1},...,U_{1C_{a_1}}$ to $X_1$ in Fig. \ref{fig:in-hk-o}, and let $P_1^*(X_1|U_{1C_1},...,U_{1C_{a_1}})=\sum_{U_{11}}P(U_{11})P^*(X_1|U_{11},U_{1C_1},...,U_{1C_{a_1}})$ be the conditional distribution of random variables $X_1,U_{1C_1},...,U_{1C_{a_1}}$ in Fig. \ref{fig:in-hk}. We can see that synthesized MACs for other receivers are not affected with this setting since $U_{11}$ plays no part in them. Thus, the achievable rate regions of other receivers' synthesized MACs remain the same. Note that deterministic mapping $P^*$ and ARV $U_{11}$ are no longer needed in this design.
  
  Now let us discuss the achievable rates from Receiver 1's point of view. Denote Sender 1's common messages as a whole by $\mathbf{U}_{1c_1}$ with rate $R_1^{c_1}$, and other senders' common messages which are intended for Receiver 1 by $\mathbf{U}_{1c_{o}}$ with rate $R_1^{c_o}$. The private message rate is denoted by $R_1^p$. With the homogeneous approach in Fig. \ref{fig:in-hk-o}, the achievable rate region of Receiver 1 is
  \begin{equation}
  \label{IN-Rp-o}
  \mathcal{R}_{IN}^1(P^*) = \left\lbrace
  \begin{matrix}
  \left(
  \begin{array}{ccc}
  R_1^p\\
  R_1^{c_1}\\
  R_1^{c_o}
  \end{array}
  \right)&
  \left|
  \begin{array}{ccc}
  \begin{array}{ccc}
  R_1^p\leq I(U_{11};Y_1|\mathbf{U}_{1c_1},\mathbf{U}_{1c_{o}})\\
  R_1^{c_1}\leq I(\mathbf{U}_{1c_1};Y_1|U_{11},\mathbf{U}_{1c_{o}})\\
  R_1^{c_o}\leq I(\mathbf{U}_{1c_{o}};Y_1|U_{11},\mathbf{U}_{1c_1})\\
  R_1^p+R_1^{c_1}\leq I(U_{11},\mathbf{U}_{1c_1};Y_1|\mathbf{U}_{1c_{o}})\\
  R_1^p+R_1^{c_o}\leq I(U_{11},\mathbf{U}_{1c_o};Y_1|\mathbf{U}_{1c_1})\\
  R_1^{c_1}+R_1^{c_o}\leq I(U_{1c_{o}},\mathbf{U}_{1c_1};Y_1|\mathbf{U}_{11})\\
  R_1^p+R_1^{c_1}+R_1^{c_o}\leq I(U_{11},\mathbf{U}_{1c_1},\mathbf{U}_{1c_{o}};Y_1)
  \end{array}
  \end{array}\right.
  \end{matrix}
  \right\rbrace .
  \end{equation}
  With the heterogeneous approach in Fig. \ref{fig:in-hk}, the achievable rate region of Receiver 1 becomes
  \begin{equation}
  \label{IN-Rp-c}
  \mathcal{R}_{IN}^{1'}(P_1^*) = \left\lbrace
  \begin{matrix}
  \left(
  \begin{array}{ccc}
  R_1^p\\
  R_1^{c_1}\\
  R_1^{c_o}
  \end{array}
  \right)&
  \left|
  \begin{array}{ccc}
  \begin{array}{ccc}
  R_1^{c_o}\leq I(\mathbf{U}_{1c_{o}};Y_1|X_1)\\
  R_1^p+R_1^{c_1}\leq I(X_1;Y_1|\mathbf{U}_{1c_{o}})\\
  R_1^p+R_1^{c_1}+R_1^{c_o}\leq I(X_1,\mathbf{U}_{1c_{o}};Y_1)
  \end{array}
  \end{array}\right.
  \end{matrix}
  \right\rbrace .
  \end{equation}
  
  Since $(U_{11},\mathbf{U}_{1c_1})\rightarrow X_1$ is a deterministic mapping, we can readily see that upper bounds for $R_1^{c_o}$, $R_1=R_1^p+R_1^{c_1}$ and $R_1^{all}=R_1^p+R_1^{c_1}+R_1^{c_o}$ are invariant with the heterogeneous approach. Thus, if we are interested in the overall rate between the user pair of Sender 1 and Receiver 1 rather than each component message rate, the heterogeneous approach can achieve the same or even a larger rate region than the homogeneous approach for a given joint distribution. 
  
  Similar to the 2-user DM-IC case, when we apply polar codes to realize the heterogeneous scheme, the design of fully-joint decoders is a problem as a sender's common messages must be decoded before its private message. Now consider using the proposed partially-joint decoding scheme. With the Type I decoding order, all common messages intended for Receiver 1 are jointly decoded before the private message. The achievable rate region is
  \begin{equation}
  \mathcal{R}_{IN}^{ParI}(P_1^*) = \left\lbrace
  \begin{matrix}
  \left(
  \begin{array}{ccc}
  R_1^p\\
  R_1^{c_1}\\
  R_1^{c_o}
  \end{array}
  \right)&
  \left|
  \begin{array}{ccc}
  \begin{array}{ccc}
  R_1^p\leq I(X_1;Y_1|\mathbf{U}_{1c_1},\mathbf{U}_{1c_{o}})\\
  R_1^{c_1}\leq I(\mathbf{U}_{1c_1};Y_1|\mathbf{U}_{1c_{o}})\\
  R_1^{c_o}\leq I(\mathbf{U}_{1c_{o}};Y_1|\mathbf{U}_{1c_1})\\
  R_1^{c_1}+R_1^{c_o}\leq I(\mathbf{U}_{1c_1},\mathbf{U}_{1c_{o}};Y_1)
  \end{array}
  \end{array}\right.
  \end{matrix}
  \right\rbrace .
  \end{equation}
  With the Type II decoding order, Sender 1's common messages are decoded first, and then the private message and other senders' common messages are jointly decoded. The achievable rate region is
  \begin{equation}
  \mathcal{R}_{IN}^{ParII}(P_1^*) = \left\lbrace
  \begin{matrix}
  \left(
  \begin{array}{ccc}
  R_1^p\\
  R_1^{c_1}\\
  R_1^{c_o}
  \end{array}
  \right)&
  \left|
  \begin{array}{ccc}
  \begin{array}{ccc}
  R_1^{c_1}\leq I(\mathbf{U}_{1c_1};Y_1)\\
  R_1^p\leq I(X_1;Y_1|\mathbf{U}_{1c_1},\mathbf{U}_{1c_{o}})\\
  R_1^{c_o}\leq I(\mathbf{U}_{1c_{o}};Y_1|X_1)\\
  R_1^p+R_1^{c_o}\leq I(X_1,\mathbf{U}_{1c_{o}};Y_1|\mathbf{U}_{1c_1})
  \end{array}
  \end{array}\right.
  \end{matrix}
  \right\rbrace .
  \end{equation}
  It is easy to verify that the following two regions,  $\{(R_1,R_1^{c_o}):R_1=R_1^p+R_1^{c_1},(R_1^p,R_1^{c_1},R_1^{c_o})\in\mathcal{R}_{IN}^{ParI}(P_1^*)\cup \mathcal{R}_{IN}^{ParII}(P_1^*)\}$ and $\{(R_1,R_1^{c_o}):R_1=R_1^p+R_1^{c_1},(R_1^p,R_1^{c_1},R_1^{c_o})\in\mathcal{R}_{IN}^{1'}(P_1^*)\}$, are equivalent. 
  
  In the above we have discussed the case when only one user pair applies heterogeneous superposition coding and partially-joint decoding. More complicated cases can be extended from this case by adding one user pair with the proposed scheme at a time. To apply polar coding, one simply needs to adopt MAC polarization with more than 2 users and follow our proposed scheme for the 2-user DM-IC. To conclude, we have the following proposition.
  \begin{proposition}
  	\label{proposition:IN}
  	The proposed heterogeneous superposition polar coding scheme with the two types of partially-joint decoding achieves the optimal rate region of DM-INs when the encoding is restricted to random coding ensembles.
  \end{proposition}
  
  \begin{remark}
  	Comparing (\ref{IN-Rp-o}) and (\ref{IN-Rp-c}) we can see that the heterogeneous approach has a much simpler expression of achievable rate region. Since we have shown that these two superposition schemes result in the same achievable rate region with respect to the overall rate between each user pair, the heterogeneous approach can serve as an useful tool for deriving simplified achievable rate regions for DM-INs.
  \end{remark}

 \section{Conclusion Remarks}
 \label{S:Conclusion}
  Based on the compact description of the Han-Kobayashi region and the coding strategy lying behind \cite{chong2008han}, we have shown that every point on the dominant faces of the Han-Kobayashi region can be achieved by polar codes in a simpler way compared to the scheme of \cite{wang2015channel}. We prove that the fully-joint decoding requirement in the Han-Kobayashi coding strategy can be loosened to partially-joint decoding, which is more friendly to polar code designs. This result reveals more insights on the roles of ARVs and coding strategies for DM-INs.

  The chaining method we used in this paper and the polar alignment technique used in \cite{wang2015channel} both make polar coding schemes lengthy, as a much larger block length is needed to achieve close-to-optimal rates. It is shown in \cite{sasoglu2011polar} that the non-universality of polar codes is a property of the successive cancellation decoding algorithm. Under ML decoding, a polar code constructed for the binary symmetric channel (BSC) universally achieves the capacity for any binary memoryless symmetric (BMS) channel. Also, as we have mentioned in Remark \ref{Remark:2}, fully-joint decoding in the heterogeneous superposition coding scheme is possible with ML decoding. This makes us wonder if there exist ML-like decoding algorithms and the corresponding code structures for polar codes which maintain universality while still enjoying low complexity. If the answer is yes, our proposed scheme may be further simplified as well as polar coding schemes for other multi-user channels.

\appendices  
\section{Proof Of Theorem \ref{theorem:HK-partial}}
\label{APPEN-A}

\begin{definition}[{\cite[p. 521]{cover2012informtaion}}]
	Let $(X_1,X_2,...,X_k)$ denote a finite collection of discrete random variables
	with some fixed joint distribution, $P_{X_1X_2...X_k}(x_1,x_2,...,x_k)$, $(x_1,x_2,
	...,x_k)\in \mathcal{X}_1\times\mathcal{X}_2\times...\times\mathcal{X}_k$. Let $S$ denote an ordered subset of these	random variables and consider $N$ independent copies of $S$. Thus,
	\begin{align*}
	\mathrm{Pr}(\mathbf{S}^N=\mathbf{s}^N)=\prod_{i=1}^{N}\mathrm{Pr}(S_i=s_i),~~\mathbf{s}^N\in \mathcal{S}^N.
	\end{align*}
	The set $\mathcal{T}_{\epsilon}^{(N)}$ of $\epsilon$-typical $N$-sequences $(\mathbf{x}_1^N,\mathbf{x}_2^N,...,\mathbf{x}_k^N)$ is defined as
	\begin{align*}
	&~~~~\mathcal{T}_{\epsilon}^{(N)}(X_1,X_2,...,X_k)\\
	&=\bigg{\{} (\mathbf{x}_1^N,\mathbf{x}_2^N,...,\mathbf{x}_k^N): \Big{|} -\frac{1}{N}\log P_{\mathbf{S}^N}(\mathbf{s}^N)-H(S) \Big{|}< \epsilon, \forall S\subseteq (X_1,X_2,...,X_k) \bigg{\}}.
	\end{align*}
\end{definition}

\textbf{Codebook generation.} Consider a fixed $P_Q(q)P_{X_1W_1|Q}(x_1,w_1|q)P_{X_2W_2|Q}(x_2,w_2|q)$. Generate a sequence $q^{1:N}\sim \prod_{j=1}^{N}P_Q(q)$. For $k=1,2$, randomly and independently generate $2^{NR_k^c}$ codewords $w_k^{1:N}(m_{kc})$, $m_{kc}\in [1:2^{NR_k^c}]$, each according to $\prod_{j=1}^{N}P_{W_k|Q}(w_k^j|q^j)$. For each $m_{kc}$, randomly and conditionally independently generate $2^{NR_k^p}$ codewords $x_k^{1:N}(m_{kc},m_{kp})$, $m_{kp}\in [1:2^{NR_k^p}]$, each according to $\prod_{j=1}^{N}P_{X_k|W_kQ}(x_k^j|w_k^j(m_{kc}),q^j)$.

\textbf{Encoding.} To send $m_k=(m_{kc},m_{kp})$, Sender $k$ ($k=1,2$) transmits $x_k^{1:N}(m_{kc},m_{kp})$.

\textbf{Decoding.} 

In the Type I partially-joint decoding, Receiver $k$ ($k=1,2$) decodes in the following two steps: 
\begin{itemize}
	\item (Simultaneous decoding for two senders' common messages) The decoder declares that $(\hat{m}_{1c},\hat{m}_{2c})$ ̂is sent if it is the unique message pair such that $$\big{(}q^{1:N},w_1^{1:N}(\hat{m}_{1c}),w_2^{1:N}(\hat{m}_{2c}),y_k^{1:N}\big{)}\in \mathcal{T}_{\epsilon}^{(N)};$$ otherwise it declares an error.
	\item (Private message decoding) If such a $(\hat{m}_{1c},\hat{m}_{2c})$ is found, the decoder finds the unique $\hat{m}_{kp}$ such that $$\big{(}q^{1:N},w_1^{1:N}(\hat{m}_{1c}),w_2^{1:N}(\hat{m}_{2c}),x_k^{1:N}(\hat{m}_{kc},\hat{m}_{kp}),y_k^{1:N}\big{)}\in \mathcal{T}_{\epsilon}^{(N)};$$ otherwise it declares an error.
\end{itemize}

In the Type II partially-joint decoding, Receiver $k$ decodes in the following two steps: 
\begin{itemize}
	\item (Intended common message decoding) The decoder declares that $\hat{m}_{kc}$ ̂is sent if it is the unique message such that $$\big{(}q^{1:N},w_k^{1:N}(\hat{m}_{kc}),y_k^{1:N}\big{)}\in \mathcal{T}_{\epsilon}^{(N)};$$ otherwise it declares an error.
	\item (Simultaneous decoding for the unintended common message and the private message) If such a $\hat{m}_{kc}$ is found, the decoder finds the unique $(\hat{m}_{k'c},\hat{m}_{kp})$ such that $$\big{(}q^{1:N},w_k^{1:N}(\hat{m}_{kc}),w_{k'}^{1:N}(\hat{m}_{k'c}),x_k^{1:N}(\hat{m}_{kc},\hat{m}_{kp}),y_k^{1:N}\big{)}\in \mathcal{T}_{\epsilon}^{(N)},$$ where $k'=\mod(k,2)+1$; otherwise it declares an error.
\end{itemize}

\textbf{Error analysis.} 
First we consider Type I. Assume that message pair $((1, 1), (1, 1))$ is sent and Receiver 1 applies the Type I decoding. Define the following error events
\begin{align*}
\mathcal{E}^{(I)}_{10}&\triangleq \{(q^{1:N},w_1^{1:N}(1),w_2^{1:N}(1),x_1^{1:N}(1,1),y_1^{1:N})\notin \mathcal{T}_{\epsilon}^{(N)}\}, \\
\mathcal{E}^{(I)}_{11}&\triangleq \{(q^{1:N},w_1^{1:N}(m_{1c}),w_2^{1:N}(1),y_1^{1:N})\in \mathcal{T}_{\epsilon}^{(N)} \text{ for some } m_{1c}\neq 1\}, \\
\mathcal{E}^{(I)}_{12}&\triangleq \{(q^{1:N},w_1^{1:N}(1),w_2^{1:N}(m_{2c}),y_1^{1:N})\in \mathcal{T}_{\epsilon}^{(N)} \text{ for some } m_{2c}\neq 1\}, \\
\mathcal{E}^{(I)}_{13}&\triangleq \{(q^{1:N},w_1^{1:N}(m_{1c}),w_2^{1:N}(m_{2c}),y_1^{1:N})\in \mathcal{T}_{\epsilon}^{(N)} \text{ for some } m_{1c}\neq 1,m_{2c}\neq 1\}, \\
\mathcal{E}^{(I)}_{14}&\triangleq \{(q^{1:N},w_1^{1:N}(1),w_2^{1:N}(1),x_1^{1:N}(1,m_{1p}),y_1^{1:N})\in \mathcal{T}_{\epsilon}^{(N)} \text{ for some } m_{1p}\neq 1\}, \\
\mathcal{E}^{(I)}_{15}&\triangleq \{(q^{1:N},w_1^{1:N}(1),w_2^{1:N}(m_{2c}),x_1^{1:N}(1,m_{1p}),y_1^{1:N})\in \mathcal{T}_{\epsilon}^{(N)} \text{ for some } m_{2c}\neq 1,m_{1p}\neq 1\}. 
\end{align*}
The average probability of error for Receiver 1  can be upper bounded as
\begin{align}
\mathrm{P}(\mathcal{E}_1^{(I)})&\leq \mathrm{P}(\mathcal{E}_{10}^{(I)})+\mathrm{P}(\mathcal{E}_{11}^{(I)})+\mathrm{P}(\mathcal{E}_{13}^{(I)})+\mathrm{P}(\mathcal{E}_{14}^{(I)})+\mathrm{P}(\mathcal{E}_{15}^{(I)})\nonumber\\
&\leq \mathrm{P}(\mathcal{E}_{10}^{(I)})+\mathrm{P}(\mathcal{E}_{11}^{(I)})+\mathrm{P}(\mathcal{E}_{12}^{(I)})+\mathrm{P}(\mathcal{E}_{13}^{(I)})+\mathrm{P}(\mathcal{E}_{14}^{(I)}),\label{PrI-1}
\end{align}
where (\ref{PrI-1}) holds because $\mathrm{P}(\mathcal{E}_{15}^{(I)})\leq \mathrm{P}(\mathcal{E}_{12}^{(I)})$.
By the law of large numbers (LLN), $\mathrm{P}(\mathcal{E}_{10}^{(I)})$ tends to 0 as $N\rightarrow \infty$. By the packing lemma, $\mathrm{P}(\mathcal{E}_{11}^{(I)})$, $\mathrm{P}(\mathcal{E}_{12}^{(I)})$, $\mathrm{P}(\mathcal{E}_{13}^{(I)})$ and $\mathrm{P}(\mathcal{E}_{14}^{(I)})$ tend to 0 as $N\rightarrow \infty$ if the conditions 
\begin{equation}
\label{RandCo-1}
\begin{aligned}
R_1^c&\leq I(W_1;Y_1|W_2Q)\\
R_2^c&\leq I(W_2;Y_1|W_1Q)\\
R_1^c+R_2^c&\leq I(W_1W_2;Y_1|Q)\\
R_1^p&\leq I(X_1;Y_1|W_1W_2Q)
\end{aligned}
\end{equation}
are satisfied, respectively. The rate constraints when Receiver 2 applies Type I decoding are similar by swapping subscripts 1 and 2 in (\ref{RandCo-1}).

Next we consider Type II. We also assume that message pair $((1, 1), (1, 1))$ is sent and Receiver 1 applies the Type II decoding. Define the following error events
\begin{align*}
\mathcal{E}^{(II)}_{10}&\triangleq \{(q^{1:N},w_1^{1:N}(1),w_2^{1:N}(1),x_1^{1:N}(1,1),y_1^{1:N})\notin \mathcal{T}_{\epsilon}^{(N)}\}, \\
\mathcal{E}^{(II)}_{11}&\triangleq \{(q^{1:N},w_1^{1:N}(m_{1c}),y_1^{1:N})\in \mathcal{T}_{\epsilon}^{(N)} \text{ for some } m_{1c}\neq 1\}, \\
\mathcal{E}^{(II)}_{12}&\triangleq \{(q^{1:N},w_1^{1:N}(1),w_2^{1:N}(1),x_1^{1:N}(1,m_{1p}),y_1^{1:N})\in \mathcal{T}_{\epsilon}^{(N)} \text{ for some } m_{1p}\neq 1\}, \\
\mathcal{E}^{(II)}_{13}&\triangleq \{(q^{1:N},w_1^{1:N}(1),w_2^{1:N}(m_{2c}),x_1^{1:N}(1,m_{1p}),y_1^{1:N})\in \mathcal{T}_{\epsilon}^{(N)} \text{ for some } m_{2c}\neq 1,m_{1p}\neq 1\}. 
\end{align*}
The average probability of error for Receiver 1  can be upper bounded as
\begin{equation}
\mathrm{P}(\mathcal{E}_1^{(II)})\leq \mathrm{P}(\mathcal{E}_{10}^{(II)})+\mathrm{P}(\mathcal{E}_{11}^{(II)})+\mathrm{P}(\mathcal{E}_{12}^{(II)})+\mathrm{P}(\mathcal{E}_{13}^{(II)}).
\end{equation}
Similarly, by the LLN, $\mathrm{P}(\mathcal{E}_{10}^{(II)})$ tends to 0 as $N\rightarrow \infty$. By the packing lemma, $\mathrm{P}(\mathcal{E}_{11}^{(II)})$, $\mathrm{P}(\mathcal{E}_{12}^{(II)})$, and $\mathrm{P}(\mathcal{E}_{13}^{(II)})$ tend to 0 as $N\rightarrow \infty$ if the conditions 
\begin{equation}
\label{RandCo-2}
\begin{aligned}
R_1^c&\leq I(W_1;Y_1|Q)\\
R_1^p&\leq I(X_1;Y_1|W_1W_2Q)\\
R_1^p+R_2^c&\leq I(X_1W_2;Y_1|W_1Q)
\end{aligned}
\end{equation}
are satisfied, respectively. The rate constraints when Receiver 2 applies the Type II decoding are similar by swapping subscripts 1 and 2 in (\ref{RandCo-2}).

Suppose both receivers adopt the Type I decoding. From (\ref{RandCo-1}) and its counterpart for Receiver 2 we know that the achievable rate region is
\begin{equation}
\label{RandCo-T1}
\mathcal{R}_{Par}^1(P_1^*)=
\left\lbrace
\begin{matrix}
\left(
\begin{array}{ccc}
R_1^c\\
R_2^c\\
R_1^p\\
R_2^p
\end{array}
\right)&
\left|
\begin{array}{ccc}
\begin{array}{ccc}
R_1^c\leq \min\{I(W_1;Y_1|W_2Q),I(W_1;Y_2|W_2Q)\}\\
R_2^c\leq \min\{I(W_2;Y_1|W_1Q),I(W_2;Y_2|W_1Q)\}\\
R_1^c+R_2^c\leq \min\{I(W_1W_2;Y_1|Q),I(W_1W_2;Y_2|Q)\}\\
R_1^p\leq I(X_1;Y_1|W_1W_2Q)\\
R_2^p\leq I(X_2;Y_2|W_1W_2Q)
\end{array}
\end{array}\right.
\end{matrix}
\right\rbrace .
\end{equation}
Now suppose Receiver 1 uses Type I while Receiver 2 adopts Type II. From (\ref{RandCo-1}) and the counterpart of (\ref{RandCo-2}) for Receiver 2 we have
\begin{equation}
\label{RandCo-T2}
\mathcal{R}_{Par}^2(P_1^*)=
\left\lbrace
\begin{matrix}
\left(
\begin{array}{ccc}
R_1^c\\
R_2^c\\
R_1^p\\
R_2^p
\end{array}
\right)&
\left|
\begin{array}{ccc}
\begin{array}{ccc}
R_1^c\leq I(W_1;Y_1|W_2Q)\\
R_2^c\leq \min\{I(W_2;Y_1|W_1Q),I(W_2;Y_2|Q)\}\\
R_1^c+R_2^c\leq I(W_1W_2;Y_1|Q)\\
R_1^p\leq I(X_1;Y_1|W_1W_2Q)\\
R_2^p\leq I(X_2;Y_2|W_1W_2Q)\\
R_2^p+R_1^c\leq I(X_2W_1;Y_2|W_2Q)
\end{array}
\end{array}\right.
\end{matrix}
\right\rbrace .
\end{equation}
Similarly if Receiver 2 uses Type I while Receiver 1 adopts Type II, we have
\begin{equation}
\label{RandCo-T3}
\mathcal{R}_{Par}^3(P_1^*)=
\left\lbrace
\begin{matrix}
\left(
\begin{array}{ccc}
R_1^c\\
R_2^c\\
R_1^p\\
R_2^p
\end{array}
\right)&
\left|
\begin{array}{ccc}
\begin{array}{ccc}
R_1^c\leq \min\{I(W_1;Y_2|W_2Q),I(W_1;Y_1|Q)\}\\
R_2^c\leq I(W_2;Y_2|W_1Q)\\
R_1^c+R_2^c\leq I(W_1W_2;Y_2|Q)\\
R_1^p\leq I(X_1;Y_1|W_1W_2Q)\\
R_2^p\leq I(X_2;Y_2|W_1W_2Q)\\
R_1^p+R_2^c\leq I(X_1W_2;Y_1|W_1Q)
\end{array}
\end{array}\right.
\end{matrix}
\right\rbrace .
\end{equation}

From (\ref{RandCo-T1}) we have
\begin{align*}
R_1^p+R_1^c+R_2^c&\leq \min\{I(X_1W_2;Y_1|Q), I(X_1;Y_1|W_1W_2Q)+I(W_1W_2;Y_2|Q)\},\\
R_2^p+R_1^c+R_2^c&\leq \min\{I(X_2W_1;Y_2|Q), I(X_2;Y_2|W_1W_2Q)+I(W_1W_2;Y_1|Q)\}.
\end{align*}
From (\ref{RandCo-T2}) we have
\begin{align*}
R_1^p+R_1^c+R_2^c&\leq I(X_1W_2;Y_1|Q),\\
R_2^p+R_1^c+R_2^c&\leq \min\left\{\begin{matrix}
I(X_2W_1;Y_2|Q)\\
I(X_2W_1;Y_2|W_2Q)+I(W_2;Y_1|W_1Q),\\
I(X_2;Y_2|W_1W_2Q)+I(W_1W_2;Y_1|Q)
\end{matrix}\right\}.
\end{align*}
From (\ref{RandCo-T3}) we have
\begin{align*}
R_1^p+R_1^c+R_2^c&\leq\min\left\{\begin{matrix}
I(X_1W_2;Y_1|Q)\\
I(X_1W_2;Y_1|W_1Q)+I(W_1;Y_2|W_2Q),\\
I(X_1;Y_1|W_1W_2Q)+I(W_1W_2;Y_2|Q)
\end{matrix}\right\},\\
R_2^p+R_1^c+R_2^c&\leq I(X_2W_1;Y_2|Q).
\end{align*}
Then we can obtain the following achievable rate region
\begin{equation}
\label{RandCo-A}
\mathcal{R}_{Par}(P_1^*)=
\left\lbrace
\begin{matrix}
\left(
\begin{array}{ccc}
R_1^c\\
R_2^c\\
R_1^p\\
R_2^p
\end{array}
\right)&
\left|
\begin{array}{ccc}
\begin{array}{ccc}
R_1^c\leq I(W_1;Y_1|W_2Q)\\
R_1^p\leq I(X_1;Y_1|W_1W_2Q)\\
R_1^p+R_2^c\leq I(X_1W_2;Y_1|W_1Q)\\
R_1^p+R_1^c+R_2^c\leq I(X_1W_2;Y_1|Q)\\
R_2^c\leq I(W_2;Y_2|W_1Q)\\
R_2^p\leq I(X_2;Y_2|W_1W_2Q)\\
R_2^p+R_1^c\leq I(X_2W_1;Y_2|W_2Q)\\
R_2^p+R_1^c+R_2^c\leq I(X_2W_1;Y_2|Q)
\end{array}
\end{array}\right.
\end{matrix}
\right\rbrace .
\end{equation}

Using the Fourier-Motzkin elimination we can readily show that (\ref{RandCo-A}) results in the same region as in Theorem \ref{theorem:HK} with the following two additional constraints (same as the Chong-Motani-Garg region shown in \cite[Lemma 4]{chong2008han}):
\begin{equation}
\begin{aligned}
R_1&\leq I(X_1;Y_1|W_1W_2Q)+I(X_2W_1;Y_2|W_2Q),\\
R_2&\leq I(X_2;Y_2|W_1W_2Q)+I(X_1W_2;Y_1|W_1Q).
\end{aligned}
\end{equation}
From \cite{chong2008han} we know that the Chong-Motani-Garg region is smaller than the compact Han-Kobayashi region for a $\mathcal{P}_1$ only if
\begin{align}
I(X_2W_1;Y_1|W_2Q) < I(W_1;Y_1|Q) \text{ or } I(X_1W_2;Y_2|W_1Q) < I(W_2;Y_2|Q).
\end{align}
For the former case, an intuitive interpretation is that Receiver 2 is unable to achieve the unintended common message rate of $R_1^c=I(W_1;Y_1|Q)$ even if it tries its best. In this case, Sender 1 will not transmit any common message (i.e., $W_1=\emptyset$) \cite[Problem 6.12]{el2011network}. Similarly, for the latter case, we will set $W_2=\emptyset$. Thus, these rates are still achievable with the proposed scheme. This completes the proof.

\section{Proof of Lemma \ref{lemma:typeAB}}
\label{APPEN-B}
Since $R_1+R_2=c$, $R_1+R_2=d$, $R_1+R_2=e$, $2R_1+R_2=f$ and $R_1+2R_2=g$ are the possible dominant faces of the Han-Kobayashi region, we prove Lemma 1 by deriving value ranges of common message rates for points on each of them.
\subsection{Points on $R_1+R_2=c$ and $R_1+R_2=d$}
Suppose $\mathbf{P}\in\mathcal{R}_{HK}(P_1^*)$ is a point on line 
\begin{equation}
\label{AP-C-1}
R_1+R_2=c.
\end{equation}
Let $(R_1^p,R_1^c,R_2^p,R_2^c)$ be a rate decomposition of $\mathbf{P}$. The equality of (\ref{AP-C-1}) forces those in the counterpart of (\ref{HK1-1}) for Receiver 2 and (\ref{HK1-7}) to hold. Thus,
\begin{align}
R_2^p&=I(X_2;Y_2|W_1W_2Q),\label{AP-B-B1} \\
R_1^p+R_1^c+R_2^c&= I(X_1W_2;Y_1|Q). \label{AP-B-B2}
\end{align}
From (\ref{AP-B-B2}) and (\ref{HK-1}) we have 
\begin{equation*}
R_2^c\geq I(W_2;Y_1|Q).
\end{equation*}
From (\ref{AP-B-B1}) and (\ref{HK-2}) we have 
\begin{equation}
R_2^c\leq I(W_2;Y_2|W_1Q).\label{AP-B-B3}
\end{equation}
From (\ref{AP-C-1}) and (\ref{HK-6}) we have
\begin{align*}
R_2\geq 2c-f=I(X_2;Y_2|W_1W_2Q)+I(W_1W_2;Y_1|Q)-I(W_1;Y_2|W_2Q).
\end{align*}
From  (\ref{AP-C-1}) and (\ref{HK-7}) we have
\begin{equation*}
R_2\leq g-c=I(X_2W_1;Y_2|Q)-I(W_1;Y_1|Q).
\end{equation*}
Thus,
\begin{equation*}
	I(W_1W_2;Y_1|Q)-I(W_1;Y_2|W_2Q)\leq R_2^c \leq I(W_1W_2;Y_2|Q)-I(W_1;Y_1|Q).
\end{equation*}
If $R_1+R_2=c$ is a dominant face of the Han-Kobayashi region, $c\leq d$ and $c\leq e$ must hold. From (\ref{HK-3}), (\ref{HK-4}) and (\ref{HK-5}) we have
\begin{align}
I(W_1;Y_2|W_2Q)&\geq I(W_1;Y_1|Q),\nonumber\\
I(W_1W_2;Y_2|Q)&\geq I(W_1W_2;Y_1|Q).\label{AP-B-B4}
\end{align}

1) For $\max\{I(W_2;Y_1|Q),I(W_1W_2;Y_1|Q)-I(W_1;Y_2|W_2Q)\}\leq R_2^c \leq I(W_2;Y_1|W_1Q)$ (if not null), let
\begin{align}
R_1^c&=I(W_1W_2;Y_1|Q)-R_2^c,\nonumber\\
R_1^p&=I(X_1;Y_1|W_1W_2Q).\nonumber
\end{align} 
Obviously $(R_1^c,R_2^c)\in \mathcal{R}(P_{Y_1|W_1W_2})$ and $R_1^c \leq I(W_1;Y_2|W_2Q)$. Then from (\ref{AP-B-B3}) and (\ref{AP-B-B4}) we know that $(R_1^c,R_2^c)\in \mathcal{R}(P_{Y_2|W_1W_2})$. Therefore $\mathbf{P}$ is of Type A. 

2) For $I(W_2;Y_1|W_1Q)\leq R_2^c \leq \min\{I(W_1W_2;Y_2|Q)-I(W_1;Y_1|Q),I(W_2;Y_2|W_1Q)\}$ (if not null), let
\begin{align*}
R_1^c&=I(W_1;Y_1|Q),\\
R_1^p&=I(X_1W_2;Y_1|W_1Q)-R_2^c.
\end{align*} 
In this case, $\mathbf{P}$ belongs to Type B.

For a point $\mathbf{P}\in\mathcal{R}_{HK}(P_1^*)$ on line  $R_1+R_2=d$, the analysis is similar.

\subsection{Points on $R_1+R_2=e$}
Suppose $\mathbf{P}\in\mathcal{R}_{HK}(P_1^*)$ is a point on line 
\begin{equation}
\label{AP-C-4}
R_1+R_2=e.
\end{equation}
Let $(R_1^p,R_1^c,R_2^p,R_2^c)$ be a rate decomposition of $\mathbf{P}$. The equality of (\ref{AP-C-4}) forces those in (\ref{HK1-5}) and its counterpart for Receiver 2 to hold. Thus,
\begin{align*}
R_1^p+R_2^c&= I(X_1W_2;Y_1|W_1Q),\\
R_2^p+R_1^c&= I(X_2W_1;Y_2|W_2Q).
\end{align*}
Then from (\ref{HK1-1}), (\ref{HK1-7}) and their counterparts for Receiver 2, we have
\begin{align}
I(W_1;Y_2|W_2Q)\leq R_1^c&\leq I(W_1;Y_1|Q),\label{APB-1}\\
I(W_2;Y_1|W_1Q)\leq R_2^c&\leq I(W_2;Y_2|Q).\label{APB-2}
\end{align}
From (\ref{HK-6}) and (\ref{AP-C-4}) we have
\begin{align*}
R_1^p+R_1^c&\leq I(X_1;Y_1|W_1W_2Q)+I(W_1;Y_1|Q),\\
R_2^p+R_2^c&\geq I(X_2W_1;Y_2|W_2Q)+I(W_2;Y_1|W_1Q)-I(W_1;Y_1|Q).
\end{align*}
From (\ref{HK-7}) and (\ref{AP-C-4}) we have
\begin{align*}
R_1^p+R_1^c&\geq I(X_1W_2;Y_1|W_1Q)+I(W_1;Y_2|W_2Q)-I(W_2;Y_2|Q),\\
R_2^p+R_2^c&\leq I(X_2;Y_2|W_1W_2Q)+I(W_2;Y_2|Q).
\end{align*}

1) If $I(X_2;Y_2|W_1W_2Q)+I(W_2;Y_1|W_1Q)\leq \mathbf{P}(2) \leq I(X_2;Y_2|W_1W_2Q)+I(W_2;Y_2|Q)$, let $R_2^p=I(X_2;Y_2|W_1W_2Q)$. Then
\begin{align*}
R_2^c&=\mathbf{P}(2)-I(X_2;Y_2|W_1W_2Q),\\
R_1^c&=I(W_1;Y_2|W_2Q),\\
R_1^p&=I(X_1W_2;Y_1|W_1Q)-R_2^c.
\end{align*}
From (\ref{APB-1}) we know that $R_1^c\leq I(W_1;Y_1|Q)$. Thus, $\mathbf{P}$ belongs to Type B.

2) If $I(X_2W_1;Y_2|W_2Q)+I(W_2;Y_1|W_1Q)-I(W_1;Y_1|Q)\leq \mathbf{P}(2)<I(X_2;Y_2|W_1W_2Q)+I(W_2;Y_1|W_1Q)$, let $R_1^p=I(X_1;Y_1|W_1W_2Q)$. Then
\begin{align*}
R_2^c&=I(W_2;Y_1|W_1Q),\\
R_2^p&=\mathbf{P}(2)-I(W_2;Y_1|W_1Q),\\
R_1^c&=I(X_2W_1;Y_2|W_2Q)+I(W_2;Y_1|W_1Q)-\mathbf{P}(2).
\end{align*}
From (\ref{APB-2}) we know that $R_2^c\leq I(W_2;Y_2|Q)$. Thus, $\mathbf{P}$ belongs to Type B.

\subsection{Points on $2R_1+R_2=f$ and $R_1+2R_2=g$}
Suppose $\mathbf{P}\in\mathcal{R}_{HK}(P_1^*)$ is a point on line 
\begin{equation}
\label{AP-C-2}
2R_1+R_2=f.
\end{equation}
Let $(R_1^p,R_1^c,R_2^p,R_2^c)$ be a rate decomposition of $\mathbf{P}$. The equality of (\ref{AP-C-2}) forces those in (\ref{HK1-1}), (\ref{HK1-7}) and the counterpart of (\ref{HK1-5}) to hold. Thus,
\begin{align}
R_1^p&=I(X_1;Y_1|W_1W_2Q),\label{AP-C-3}\\
R_1^c+R_2^p&=I(X_2W_1;Y_2|W_2Q),\\
R_1^c+R_2^c&=I(W_1W_2;Y_1|Q).
\end{align}
Then we obtain from (\ref{HK1-4}) that
\begin{equation}
R_1^c\leq I(W_1;Y_1|W_2Q).
\end{equation}
From (\ref{HK-3})--(\ref{HK-5}), (\ref{AP-C-2}) and (\ref{AP-C-3}) we have
\begin{align*}
R_1^c\geq \max\left\{\begin{matrix}
I(W_1;Y_2|W_2Q) \\
I(W_1W_2;Y_1|Q)-I(W_2;Y_2|Q)\\
I(W_1;Y_1|Q)
\end{matrix}\right\}.
\end{align*}
Thus,
\begin{align*}
R_2^c\leq \min\left\{\begin{matrix}
I(W_1W_2;Y_1|Q)-I(W_1;Y_2|W_2Q) \\
I(W_2;Y_2|Q)\\
I(W_2;Y_1|W_1Q)
\end{matrix}\right\}.
\end{align*}
We can see that $(R_1^c,R_2^c)\in \mathcal{R}(P_{Y_1|W_1W_2})$ and $R_2^c\leq I(W_2;Y_2|Q)$. Thus, $\mathbf{P}$ belongs to Type B.

For a point $\mathbf{P}\in\mathcal{R}_{HK}(P_1^*)$ on line  $R_1+2R_2=g$, the analysis is similar.

Now we have completed the proof.

\section{Proof Of Lemma \ref{lemma.I1p}}
\label{APPEN-C}
Define
\begin{equation}
\mathcal{H}^{(N)}_{S_{U_1}|Y_1W_1}\triangleq \big{\{}j\in [N]:H(S^{f_1(j)}|Y_1^{1:N},U_1^{'1:N}, S^{1:f_1(j)-1})\geq \log_2({q_{X_1}})-\delta_N \big{\}},
\end{equation}
and let $\mathcal{B}^{(N)}_{S_{U_1}|Y_1W_1}\triangleq (\mathcal{H}^{(N)}_{S_{U_1}|Y_1W_1}\cup \mathcal{L}^{(N)}_{S_{U_1}|Y_1W_1})^C$. Then we have
\begin{align*}
\frac{1}{N}|\mathcal{I}_{1p}|&=\frac{1}{N}|\mathcal{H}^{(N)}_{X_1|W_1}\cap (\mathcal{H}^{(N)}_{S_{U_1}|Y_1W_1}\cup \mathcal{B}^{(N)}_{S_{U_1}|Y_1W_1})^C|\\
&=\frac{1}{N}|\mathcal{H}^{(N)}_{X_1|W_1}\cap (\mathcal{H}^{(N)}_{S_{U_1}|Y_1W_1})^C\cap (\mathcal{B}^{(N)}_{S_{U_1}|Y_1W_1})^C|\\
&\geq \frac{1}{N}|\mathcal{H}^{(N)}_{X_1|W_1}\cap (\mathcal{H}^{(N)}_{S_{U_1}|Y_1W_1})^C|-\frac{1}{N}|\mathcal{B}^{(N)}_{S_{U_1}|Y_1W_1}|\\
&=\frac{1}{N}|\mathcal{H}^{(N)}_{X_1|W_1}|-\frac{1}{N}|\mathcal{H}^{(N)}_{S_{U_1}|Y_1W_1}|-\frac{1}{N}|\mathcal{B}^{(N)}_{S_{U_1}|Y_1W_1}|.
\end{align*}
From (\ref{PolarRate}) we have 
\begin{equation*}
\lim\limits_{N\rightarrow \infty}\frac{1}{N}|\mathcal{H}^{(N)}_{X_1|W_1}|=H_{q_{X_1}}(X_1|W_1).
\end{equation*}
From \cite[Lemma 1]{chou2015keygen} we have
\begin{equation*}
\lim\limits_{N\rightarrow \infty}\frac{1}{N}|\mathcal{B}^{(N)}_{S_{U_1}|Y_1W_1}|=0.
\end{equation*}
From (\ref{MACRate}) we have 
\begin{align}
&~~~\lim\limits_{N\rightarrow \infty}\frac{1}{N}|\mathcal{H}^{(N)}_{S_{U_1}|Y_1W_1}|\nonumber\\
&=\lim\limits_{N\rightarrow \infty}\frac{1}{N}\sum_{j\in \mathcal{S}_{U_1}}H_{q_{X_1}}(S^j|Y_1^{1:N},W_1^{1:N},S^{1:j-1})\nonumber\\
&=\lim\limits_{N\rightarrow \infty}\big{(}\frac{1}{N}H_{q_{X_1}}(S^{1:2N}|Y_1^{1:N},W_1^{1:N})-\frac{1}{N}\sum_{j\in \mathcal{S}_{U'_2}}H_{q_{X_1}}(S^j|Y_1^{1:N},W_1^{1:N},S^{1:j-1})\big{)}\nonumber\\
&=\lim\limits_{N\rightarrow \infty}\big{(}\frac{1}{N}H_{q_{X_1}}(S^{1:2N},W_1^{1:N}|Y_1^{1:N})-\frac{1}{N}H_{q_{X_1}}(W_1^{1:N}|Y_1^{1:N})\nonumber\\
&~~~~~~~~-\frac{1}{N}\sum_{j\in \mathcal{S}_{U'_2}}H_{q_{X_1}}(S^j|Y_1^{1:N},W_1^{1:N},S^{1:j-1})\big{)}\nonumber\\
&=\lim\limits_{N\rightarrow \infty}\frac{1}{N}H_{q_{X_1}}(S^{1:2N}|Y_1^{1:N})+\lim\limits_{N\rightarrow \infty}\frac{1}{N}H_{q_{X_1}}(W_1^{1:N}|Y_1^{1:N},S^{1:2N})\nonumber\\
&~~~~~~~~-H_{q_{X_1}}(W_1|Y_1)-\lim\limits_{N\rightarrow \infty}\frac{1}{N}\sum_{j\in \mathcal{S}_{U'_2}}H_{q_{X_1}}(S^j|Y_1^{1:N},W_1^{1:N},S^{1:j-1})\nonumber\\
&=H_{q_{X_1}}(X_1W_2|Y_1)-H_{q_{X_1}}(W_1|Y_1)-\lim\limits_{N\rightarrow \infty}\frac{1}{N}\sum_{j\in \mathcal{S}_{U'_2}}H_{q_{X_1}}(S^j|Y_1^{1:N},S^{1:j-1})\label{Hsu1-1}\\
&=H_{q_{X_1}}(X_1W_2|Y_1)-H_{q_{X_1}}(W_1|Y_1)-\big{(} H_{q_{X_1}}(W_2)-I(X_1W_2;Y_1)+\mathbf{\bar{P}}^1(1) \big{)}\label{Hsu1-2}\\
&=H_{q_{X_1}}(X_1)-H_{q_{X_1}}(W_1|Y_1)-\mathbf{\bar{P}}^1(1),\nonumber
\end{align}
where (\ref{Hsu1-1}) holds because $H_{q_{X_1}}(W_1^{1:N}|Y_1^{1:N},S^{1:2N})=0$, and (\ref{Hsu1-2}) holds by
\begin{align*}
\lim\limits_{N\rightarrow \infty}\frac{1}{N}\sum_{j\in \mathcal{S}_{U'_2}}H_{q_{X_1}}(S^j|Y_1^{1:N},S^{1:j-1})&=H_{q_{X_1}}(W_2)-\mathbf{\bar{P}}^1(2)\\
&=H_{q_{X_1}}(W_2)-\big{(}I(X_1W_2;Y_1)-\mathbf{\bar{P}}^1(1)\big{)}.
\end{align*}
Thus,
\begin{align}
\lim\limits_{N\rightarrow \infty}\frac{1}{N}|\mathcal{I}_{1p}|&=H_{q_{X_1}}(X_1|W_1)-\big{(}H_{q_{X_1}}(X_1)-H_{q_{X_1}}(W_1|Y_1)-\mathbf{\bar{P}}^1(1)\big{)}\nonumber\\
&=H_{q_{X_1}}(X_1W_1)-H_{q_{X_1}}(W_1)-H_{q_{X_1}}(X_1)+H_{q_{X_1}}(W_1|Y_1)+\mathbf{\bar{P}}^1(1)\nonumber\\
&=\mathbf{\bar{P}}^1(1)-I(W_1;Y_1),\label{I1p}
\end{align}
where (\ref{I1p}) holds because $H_{q_{X_1}}(X_1W_1)=H_{q_{X_1}}(X_1)$.

\section{Proof Of Lemma \ref{lemma.TVD}}
\label{APPEN-D}
We drop the subscript $(\cdot)_i$ for simplicity here as the analysis for any $i$ is the same. Since $\mathbf{G}_N$ is an invertible mapping, by the chain rule for the Kullback-Leibler divergence we have
\begin{align}
\mathbb{D}(P_{W_1^{1:N}}||Q_{W_1^{1:N}})&=\mathbb{D}(P_{U_1^{'1:N}}||Q_{U_1^{'1:N}})\nonumber\\
&=\sum_{j=1}^{N}\mathbb{D}(P_{U_1^{'j}|U_1^{'1:j-1}}||Q_{U_1^{'j}|U_1^{'1:j-1}}) \nonumber\\
&=\sum_{j\in\mathcal{H}_{W_1}^{(N)}}\mathbb{D}(P_{U_1^{'j}|U_1^{'1:j-1}}||Q_{U_1^{'j}|U_1^{'1:j-1}}) \label{AP-B-1}\\
&=\sum_{j\in\mathcal{H}_{W_1}^{(N)}}\big{(} \log_2(q_{W_1})-H(U_1^{'j}|U_1^{'1:j-1})  \big{)} \label{AP-B-2}\\
&\leq N\delta_N, \label{AP-B-3}
\end{align}
where (\ref{AP-B-1}) holds by our common message encoding scheme, (\ref{AP-B-2}) holds by the fact that information symbols and frozen symbols are uniformly distributed, and (\ref{AP-B-3}) holds by the definition of set $\mathcal{H}_{W_1}^{(N)}$. Similarly,
\begin{align}
\mathbb{D}(P_{U_1^{1:N}|W_1^{1:N}}||Q_{U_1^{1:N}|W_1^{1:N}}) 
&=\sum_{j=1}^{N}\mathbb{D}(P_{U_1^{j}|U_1^{1:j-1}W_1^{1:N}}||Q_{U_1^{j}|U_1^{1:j-1}W_1^{1:N}})\nonumber\\
&=\sum_{j\in\mathcal{H}_{X_1|W_1}^{(N)}}\mathbb{D}(P_{U_1^{j}|U_1^{1:j-1}W_1^{1:N}}||Q_{U_1^{j}|U_1^{1:j-1}W_1^{1:N}}) \nonumber\\
&=\sum_{j\in\mathcal{H}_{X_1|W_1}^{(N)}}\big{(} \log_2(q_{X_1})-H_{q_{X_1}}(U_1^{j}|U_1^{1:j-1}W_1^{1:N})  \big{)} \nonumber\\
&\leq N\delta_N. \nonumber
\end{align}
Then by the chain rule for the Kullback-Leibler divergence we have
\begin{align}
\mathbb{D}(P_{W_1^{1:N}X_1^{1:N}}||Q_{W_1^{1:N}X_1^{1:N}})
&=\mathbb{D}(P_{W_1^{1:N}U_1^{1:N}}||Q_{W_1^{1:N}U_1^{1:N}})\nonumber\\
&=\mathbb{D}(P_{U_1^{1:N}|W_1^{1:N}}||Q_{U_1^{1:N}|W_1^{1:N}})+\mathbb{D}(P_{W_1^{1:N}}||Q_{W_1^{1:N}}) \label{AP-B-4}\\
&\leq 2N\delta_N. \label{AP-B-5}
\end{align}
Similarly,
\begin{align}
&~~~~\mathbb{D}(P_{W_2^{1:N}X_2^{1:N}}||Q_{W_2^{1:N}X_2^{1:N}})\leq 2N\delta_N. \label{AP-B-6}
\end{align}
Then we have
\begin{align}
&~~\parallel P_{W_1^{1:N}W_2^{1:N}X_1^{1:N}X_2^{1:N}}-Q_{W_1^{1:N}W_2^{1:N}X_1^{1:N}X_2^{1:N}}\parallel \nonumber\\
&= \parallel P_{W_1^{1:N}X_1^{1:N}}P_{W_2^{1:N}X_2^{1:N}}-Q_{W_1^{1:N}X_1^{1:N}}Q_{W_2^{1:N}X_2^{1:N}}\parallel \nonumber\\
&\leq \parallel  P_{W_1^{1:N}X_1^{1:N}}P_{W_2^{1:N}X_2^{1:N}}-Q_{W_1^{1:N}X_1^{1:N}}P_{W_2^{1:N}X_2^{1:N}}\parallel \nonumber\\
&~~~~+ \parallel Q_{W_1^{1:N}X_1^{1:N}}P_{W_2^{1:N}X_2^{1:N}}-Q_{W_1^{1:N}X_1^{1:N}}Q_{W_2^{1:N}X_2^{1:N}}\parallel \label{VD-1}\\
&= \parallel P_{W_1^{1:N}X_1^{1:N}}-Q_{W_1^{1:N}X_1^{1:N}}\parallel + \parallel P_{W_2^{1:N}X_2^{1:N}}-Q_{W_2^{1:N}X_2^{1:N}}\parallel \label{VD-2}\\
&\leq 4\sqrt{\log 2}\sqrt{N\delta_N},\label{VD-3}
\end{align}
where (\ref{VD-1}) holds by the triangle inequality, (\ref{VD-2}) holds by \cite[Lemma 17]{cuff2009communication}, and (\ref{VD-3}) holds by (\ref{AP-B-5}), (\ref{AP-B-6}) and Pinsker's inequality.

Since $P_{Y_1^{1:N}Y_2^{1:N}|X_1^{1:N}X_2^{1:N}}=Q_{Y_1^{1:N}Y_2^{1:N}|X_1^{1:N}X_2^{1:N}}$, by \cite[Lemma 17]{cuff2009communication} we have
\begin{align}
&~~\parallel P_{W_1^{1:N}W_2^{1:N}X_1^{1:N}X_2^{1:N}Y_1^{1:N}Y_2^{1:N}}-Q_{W_1^{1:N}W_2^{1:N}X_1^{1:N}X_2^{1:N}Y_1^{1:N}Y_2^{1:N}}\parallel \nonumber\\
&=\parallel P_{W_1^{1:N}W_2^{1:N}X_1^{1:N}X_2^{1:N}}-Q_{W_1^{1:N}W_2^{1:N}X_1^{1:N}X_2^{1:N}}\parallel \nonumber\\
&\leq 4\sqrt{\log 2}\sqrt{N\delta_N}.
\end{align}

\section{Proof Of Lemma \ref{lemma.EP}}
\label{APPEN-E}
To evaluate all error events in the proposed scheme, we denote the random variables drawn from the target distribution as $U_1$, $X_1$, $Y_1$, etc., those induced by our encoding scheme as $\tilde{U}_1$, $\tilde{X}_1$, $\tilde{Y}_1$, etc., and those of the decoding results as $\bar{U}_1$, $\bar{X}_1$, $\bar{Y}_1$, etc.

We first bound the error probability of a receiver with the Type I partially-joint decoding. As an example, we consider Receiver 1 in the Type A scheme. Define the following error events 
\begin{align*}
\mathcal{E}_{1,i}&\triangleq \{(\tilde{W}_1^{1:N}\tilde{W}_2^{1:N}\tilde{X}_1^{1:N}\tilde{Y}_1^{1:N})_i\neq (W_1^{1:N}W_2^{1:N}X_1^{1:N}Y_1^{1:N})\}, \\
\mathcal{E}_{W_1W_2,i-1}^{ch}&\triangleq \{(\bar{U}_1^{'chaining}\bar{U}_2^{'chaining})_{i-1}\neq (\tilde{U}_1^{'chaining}\tilde{U}_2^{'chaining})_{i-1}\}, \\
\mathcal{E}_{W_1W_2,i}&\triangleq \{(\bar{U}_1^{'1:N}\bar{U}_2^{'1:N})_{i}\neq (\tilde{U}_1^{'1:N}\tilde{U}_2^{'1:N})_{i}\},\\
\mathcal{E}_{X_1,i}&\triangleq \{(\bar{U}_1^{1:N})_{i}\neq (\tilde{U}_1^{1:N})_{i}\},
\end{align*}
where "chaining" in the superscript stands for the elements used for chaining. 

The error probability of Receiver 1 when decoding messages in Block $i$ ($1\leq i\leq K$) can be upper bounded by
\begin{align}
P^I_{e,i} &\leq P[\mathcal{E}_{X_1,i} \text{~or~} \mathcal{E}_{W_1W_2,i}] \nonumber\\
&=P[\mathcal{E}_{X_1,i} \text{~or~} \mathcal{E}_{W_1W_2,i}|\mathcal{E}_{1,i}]P[\mathcal{E}_{1,i}]+P[\mathcal{E}_{X_1,i} \text{~or~} \mathcal{E}_{W_1W_2,i}|\mathcal{E}_{1,i}^C]P[\mathcal{E}_{1,i}^C] \nonumber\\
&\leq P[\mathcal{E}_{1,i}]+P[\mathcal{E}_{X_1,i} \text{~or~} \mathcal{E}_{W_1W_2,i}|\mathcal{E}_{1,i}^C] \nonumber\\
&= P[\mathcal{E}_{1,i}]+P[\mathcal{E}_{W_1W_2,i}|\mathcal{E}_{1,i}^C]+P[\mathcal{E}_{X_1,i}|\mathcal{E}_{1,i}^C,\mathcal{E}_{W_1W_2,i}^C]P[\mathcal{E}^C_{W_1W_2,i}|\mathcal{E}_{1,i}^C] \nonumber\\
&\leq P[\mathcal{E}_{1,i}]+P[\mathcal{E}_{W_1W_2,i}|\mathcal{E}_{1,i}^C]+P[\mathcal{E}_{X_1,i}|\mathcal{E}_{1,i}^C,\mathcal{E}_{W_1W_2,i}^C]. \label{PL4-1}
\end{align}
Using optimal coupling \cite[Lemma 3.6]{aldous1983random} we have
\begin{equation}
P[\mathcal{E}_{1,i}]=\parallel P_{W_1^{1:N}W_2^{1:N}X_1^{1:N}Y_1^{1:N}}-Q_{(W_1^{1:N}W_2^{1:N}X_1^{1:N}Y_1^{1:N})_i}\parallel. \label{PL4-4}
\end{equation}
For $i\geq 2$, we have
\begin{align}
&~~~P[\mathcal{E}_{W_1W_2,i}|\mathcal{E}_{1,i}^C]\nonumber\\ &=P[\mathcal{E}_{W_1W_2,i}|\mathcal{E}_{1,i}^C,\mathcal{E}_{W_1W_2,i-1}^{ch}]P[\mathcal{E}_{W_1W_2,i-1}^{ch}]+P[\mathcal{E}_{W_1W_2,i}|\mathcal{E}_{c,i}^C,(\mathcal{E}_{W_1W_2,i-1}^{ch})^C]P[(\mathcal{E}_{W_1W_2,i-1}^{ch})^C] \nonumber\\
&\leq P[\mathcal{E}_{W_1W_2,i-1}^{ch}]+P[\mathcal{E}_{W_1W_2,i}|\mathcal{E}_{c,i}^C,(\mathcal{E}_{W_1W_2,i-1}^{ch})^C] \nonumber\\
&\leq P[\mathcal{E}_{W_1W_2,i-1}]+N\delta_N, \label{R1C-1}
\end{align}
where (\ref{R1C-1}) holds by the error probability of source polar coding \cite{arikan2010source}. Since
\begin{align}
  P[\mathcal{E}_{W_1W_2,i-1}]&= P[\mathcal{E}_{W_1W_2,i-1}|\mathcal{E}_{1,i-1}^C]P[\mathcal{E}_{1,i-1}^C]+P[\mathcal{E}_{W_1W_2,i-1}|\mathcal{E}_{1,i-1}]P[\mathcal{E}_{1,i-1}]\nonumber\\
  &\leq P[\mathcal{E}_{W_1W_2,i-1}|\mathcal{E}_{1,i-1}^C]+P[\mathcal{E}_{1,i-1}],\nonumber
\end{align}
we have 
\begin{equation*}
P[\mathcal{E}_{W_1W_2,i}|\mathcal{E}_{1,i}^C]\leq P[\mathcal{E}_{W_1W_2,i-1}|\mathcal{E}_{1,i-1}^C]+N\delta_N+4\sqrt{\log 2}\sqrt{N\delta_N}.
\end{equation*}
For $i=1$, from our chaining scheme we know that
\begin{align*}
P[\mathcal{E}_{W_1W_2,1}|\mathcal{E}_{1,1}^C]\leq N\delta_N.
\end{align*}	
Then by induction we have
\begin{equation*}
P[\mathcal{E}_{W_1W_2,i}|\mathcal{E}_{1,i}^C]\leq iN\delta_N+4(i-1)\sqrt{\log 2}\sqrt{N\delta_N}.
\end{equation*}
By the error probability of source polar coding \cite{arikan2010source} we have
\begin{align*}
P[\mathcal{E}_{X_1,i}|\mathcal{E}_{1,i}^C,\mathcal{E}_{W_1W_2,i}^C] \leq N\delta_N.
\end{align*}
Thus,
\begin{equation*}
P_{e1,i}\leq (i+1)N\delta_N+4i\sqrt{\log 2}\sqrt{N\delta_N}.
\end{equation*}
Then error probability of Receiver 1 in the overall $K$ blocks can be upper bounded by
\begin{equation}
P^I_{e}\leq \sum_{i=1}^{K}P^I_{e,i}\leq \frac{(K+1)(K+2)}{2}N\delta_N+2K(K+1)\sqrt{\log 2}\sqrt{N\delta_N}.
\end{equation}

Next we bound the error probability of a receiver with the Type II partially-joint decoding. As an example, we consider Receiver 1 in the Type B scheme. Define the following error events 
\begin{align*}
\mathcal{E}_{1,i}&\triangleq \{(\tilde{W}_1^{1:N}\tilde{W}_2^{1:N}\tilde{X}_1^{1:N}\tilde{Y}_1^{1:N})_i\neq (W_1^{1:N}W_2^{1:N}X_1^{1:N}Y_1^{1:N})\}, \\
\mathcal{E}_{W_1,i-1}^{ch}&\triangleq \{(\bar{U}_1^{'chaining})_{i-1}\neq (\tilde{U}_1^{'chaining})_{i-1}\}, \\
\mathcal{E}_{W_2,i-1}^{ch}&\triangleq \{(\bar{U}_2^{'chaining})_{i-1}\neq (\tilde{U}_2^{'chaining})_{i-1}\}, \\
\mathcal{E}_{W_1,i}&\triangleq \{(\bar{U}_1^{'1:N})_{i}\neq (\tilde{U}_1^{'1:N})_{i}\},\\
\mathcal{E}_{X_1W_2,i}&\triangleq \{(\bar{U}_1^{1:N}\bar{U}_2^{'1:N})_{i}\neq (\tilde{U}_1^{1:N}\tilde{U}_2^{'1:N})_{i}\}.
\end{align*}
Similar to (\ref{PL4-1}), the error probability of Receiver 1 when decoding messages in Block $i$ ($1\leq i\leq K$) can be upper bounded by
\begin{align}
P^{II}_{e,i} &\leq P[\mathcal{E}_{X_1W_2,i} \text{~or~} \mathcal{E}_{W_1,i}] \nonumber\\
&\leq P[\mathcal{E}_{1,i}]+P[\mathcal{E}_{W_1,i}|\mathcal{E}_{1,i}^C]+P[\mathcal{E}_{X_1W_2,i}|\mathcal{E}_{1,i}^C,\mathcal{E}_{W_1,i}^C]. \label{PL4-2}
\end{align}
Similar to the analysis for $P[\mathcal{E}_{W_1W_2,i}|\mathcal{E}_{1,i}^C]$ in the Type I case, we have
\begin{equation}
P[\mathcal{E}_{W_1,i}|\mathcal{E}_{1,i}^C]\leq iN\delta_N+4(i-1)\sqrt{\log 2}\sqrt{N\delta_N},\label{PL4-3}
\end{equation}
and
\begin{align}
P[\mathcal{E}_{W_1,i}]&\leq
P[\mathcal{E}_{W_1,i}|\mathcal{E}_{1,i}^C]P[\mathcal{E}_{1,i}^C]+P[\mathcal{E}_{1,i}]\nonumber\\
&\leq iN\delta_N+4i\sqrt{\log 2}\sqrt{N\delta_N}. \nonumber
\end{align}
For $i\geq 2$, we have
\begin{align}
&~~~P[\mathcal{E}_{X_1W_2,i}|\mathcal{E}_{1,i}^C,\mathcal{E}_{W_1,i}^C]\nonumber\\
&=P[\mathcal{E}_{X_1W_2,i}|\mathcal{E}_{1,i}^C,\mathcal{E}_{W_1,i}^C,(\mathcal{E}_{W_2,i-1}^{ch})^C]P[(\mathcal{E}_{W_2,i-1}^{ch})^C]+P[\mathcal{E}_{X_1W_2,i}|\mathcal{E}_{1,i}^C,\mathcal{E}_{W_1,i}^C,\mathcal{E}_{W_2,i-1}^{ch}]P[\mathcal{E}_{W_2,i-1}^{ch}]  \nonumber\\
&\leq P[\mathcal{E}_{X_1W_2,i}|\mathcal{E}_{1,i}^C,\mathcal{E}_{W_1,i}^C,(\mathcal{E}_{W_2,i-1}^{ch})^C]+P[\mathcal{E}_{W_2,i-1}^{ch}] \nonumber\\
&\leq N\delta_N+P[\mathcal{E}_{X_1W_2,i-1}]\nonumber\\
&=N\delta_N+P[\mathcal{E}_{X_1W_2,i-1}|\mathcal{E}_{1,i-1}^C]P[\mathcal{E}_{1,i-1}^C]+P[\mathcal{E}_{X_1W_2,i-1}|\mathcal{E}_{1,i-1}]P[\mathcal{E}_{1,i-1}]\nonumber\\
&\leq P[\mathcal{E}_{X_1W_2,i-1}|\mathcal{E}_{1,i-1}^C]+N\delta_N+4\sqrt{\log 2}\sqrt{N\delta_N}\nonumber\\
&\leq P[\mathcal{E}_{X_1W_2,i-1}|\mathcal{E}_{1,i-1}^C,\mathcal{E}_{W_1,i-1}^C]+P[\mathcal{E}_{W_1,i-1}]+N\delta_N+4\sqrt{\log 2}\sqrt{N\delta_N}\nonumber\\
&\leq P[\mathcal{E}_{X_1W_2,i-1}|\mathcal{E}_{1,i-1}^C,\mathcal{E}_{W_1,i-1}^C]+iN\delta_N+4i\sqrt{\log 2}\sqrt{N\delta_N}.\nonumber
\end{align}
For $i=1$, from our chaining scheme we have
\begin{align*}
P[\mathcal{E}_{X_1W_2,1}|\mathcal{E}_{1,1}^C,\mathcal{E}_{W_1,1}^C] \leq N\delta_N.
\end{align*}
Thus, by induction we have
\begin{equation}
P[\mathcal{E}_{X_1W_2,i}|\mathcal{E}_{1,i}^C,\mathcal{E}_{W_1,i}^C]\leq \frac{(i+2)(i-1)}{2}(N\delta_N+4\sqrt{\log 2}\sqrt{N\delta_N})+N\delta_N. \label{PL4-5}
\end{equation}
From (\ref{PL4-2}), (\ref{PL4-4}), (\ref{PL4-3}) and (\ref{PL4-5}) we have
\begin{equation*}
P^{II}_{e,i}\leq \frac{i^2+3i-2}{2}(N\delta_N+4\sqrt{\log 2}\sqrt{N\delta_N})+N\delta_N..
\end{equation*}
Then we have
\begin{equation}
P^{II}_{e}\leq \sum_{i=1}^{K}P^{II}_{e,i}\leq \frac{K(K+1)(K+5)}{6}N\delta_N+\frac{2K(K^2+6K-1)}{3}\sqrt{\log 2}\sqrt{N\delta_N}.
\end{equation}

\bibliographystyle{IEEEtran}
\bibliography{Polar_IC}

\end{document}